\documentclass[11pt]{article}
\usepackage[T1]{fontenc}
\usepackage{lmodern}
\usepackage{hyphenat}

    \usepackage{fullpage}
    \usepackage{amsmath,amssymb,amsthm}
    \usepackage{mathtools}
    \usepackage{microtype}
    \usepackage{chngcntr} 
    \usepackage[ruled,vlined,linesnumbered]{algorithm2e}
    \SetVlineSkip{0pt}           
    \SetInd{0.6em}{0.6em}        
    \usepackage{float} 
    \usepackage{enumitem}
    \setlist{leftmargin=*, itemsep=0.2em, topsep=0.2em, parsep=0em, partopsep=0em}
    \usepackage{xcolor}
    \usepackage[colorlinks=true,linkcolor=blue,citecolor=blue,urlcolor=blue]{hyperref}
    \usepackage{complexity}
    \newtheorem{theorem}{Theorem}[section]
    \newtheorem{lemma}[theorem]{Lemma}
    \newtheorem{proposition}[theorem]{Proposition}
    \newtheorem{corollary}[theorem]{Corollary}

    \theoremstyle{definition}
    \newtheorem{definition}[theorem]{Definition}

    \numberwithin{equation}{section}
    \counterwithin{figure}{section}
    \counterwithin{table}{section}
    \counterwithin{algocf}{section} 

    \newcommand{\Jobs}{\mathcal{J}}

    \newcommand{\pos}{\operatorname{pos}}
    \newcommand{\EDD}{\textsc{EDD}}
    \newcommand{\SPT}{\textsc{SPT}}
    \newcommand{\OPT}{\textsc{OPT}}

\SetAlFnt{\small}              
\SetAlgoNlRelativeSize{-1}     
\SetInd{0.6em}{0.6em}          

    \title{The \texorpdfstring{$k$}{k}-Sum Lateness Problem on a Single Machine}
    \author{Ricardo Arancibia-Castillo\thanks{Department of Mathematical Engineering, Universidad de Chile.} \and José A. Soto\footnotemark[1] \thanks{Center for Mathematical Modeling CNRS-IRL 2807, Universidad de Chile.}}
    \date{}
    
\begin{document}
    \maketitle
    
\begin{abstract}
We study a single-machine scheduling problem in which each job $j$ has a
nonnegative processing time $p_j\ge 0$ and a due date $d_j\in\mathbb{R}$. For a
non-idling schedule $S$, let $C_j(S)$ be the completion time and let
$L_j(S)=C_j(S)-d_j$ be the (possibly negative) lateness. The objective is to
minimize the sum of the $k$ largest lateness values, interpolating between
maximum lateness ($k=1$) and total lateness ($k=n$).

We prove that the decision version is weakly NP-complete when $k$ is part of the
input. For fixed $k$, we give an $O(k^2n^{k+2})$ algorithm. As a consequence, we
resolve a conjecture of Woeginger on the top-$k$ tardiness problem and obtain an
$O(n^{k+2})$ algorithm for every fixed $k$.

Our main structural result shows that there exists an optimal schedule that
admits a block--island decomposition. Outside a suitable top-$k$ set, jobs form
due-date blocks ordered by due date. Within each due-date class, the top-$k$
jobs form a suffix in lexicographic shortest-processing-time (\SPT{}) order. This structure also yields an
FPT algorithm parameterized by $D+k$, where $D$ is the number of distinct due
dates.

Independently, a standard dual representation of the top-$k$ objective reduces
the problem to a family of total-tardiness instances with uniformly shifted due
dates. For integral data, this gives a pseudopolynomial algorithm and a fully
polynomial additive approximation scheme with error at most $\varepsilon M$,
where $M=\max\{1,\max_j p_j,\max_j |d_j|\}$. The same route also gives XP
algorithms for fixed $P$ and fixed $D$, where $P$ is the number of distinct
processing times.
\end{abstract}

    
\section{Introduction}

    Single-machine scheduling with due dates captures situations in which a set of jobs must be processed on a single machine.
    Each job $j$ has a nonnegative processing time $p_j\ge 0$ and a due date $d_j\in\mathbb{R}$ \cite{Graham1979,Lawler1993,Pinedo2016}.
    A schedule is a permutation $S$ of the jobs (with no idle time). It induces a completion time $C_j(S)$ for each job.

    A standard due-date performance metric is the \emph{lateness} $L_j(S)=C_j(S)-d_j$.
    It is negative if job $j$ completes before its due date and positive if it completes after its due date.
    The related notion of \emph{tardiness} is $T_j(S)=\max\{L_j(S),0\}$, so it does not reward early completion.

    Two classical objectives are to minimize \emph{maximum lateness}, $\max_j L_j(S)$, or \emph{total lateness}, $\sum_j L_j(S)$.
    Maximum lateness is minimized by the \emph{earliest-due-date} (\EDD{}) rule, which orders jobs by nondecreasing due dates~\cite{Jackson1955}.
    Since $\sum_j L_j(S)=\sum_j C_j(S)-\sum_j d_j$, minimizing total lateness is equivalent to minimizing total completion time, achieved by the
    \emph{shortest-processing-time} (\SPT{}) dispatching rule, which orders jobs by
    nondecreasing processing times~\cite{Smith1956}.
    
    We study the objective of minimizing the sum of the $k$ largest lateness values. This provides a natural
    interpolation between maximum lateness ($k=1$) and total lateness ($k=n$).

    This objective has a natural interpretation in delivery planning. Suppose that after processing on
    the machine, job $j$ must be delivered, requiring an additional delivery time $s_j$. Setting the due date
    to $d_j=-s_j$ yields $L_j(S)=C_j(S)+s_j$, which is exactly the delivery completion time of job $j$.
    Thus, the objective minimizes the sum of the $k$ latest delivery completion times.
    
    \paragraph{The optimization problem.}
    An instance of \emph{$k$-sum lateness on one machine} consists of $n$ jobs
    $\Jobs=\{1,\dots,n\}$ with processing times $p_j\ge 0$, due dates $d_j\in\mathbb{R}$,
    and an integer $k\in\{1,\dots,n\}$. A solution is a non-idling schedule $S$.
    With completion times $C_j(S)$ and lateness $L_j(S)=C_j(S)-d_j$, let
    $L_{(1)}(S)\ge \cdots \ge L_{(n)}(S)$ be the lateness values in non-increasing order and define
    \[
    \lambda_k(S)=\sum_{t=1}^k L_{(t)}(S).
    \]
    We seek $\min_S\lambda_k(S)$, which interpolates between maximum lateness ($k=1$) and total lateness ($k=n$).
    The objective is a special case of an ordered weighted averaging criterion and is closely related to discrete conditional value-at-risk (CVaR), motivating the dual formulation in Section~\ref{sec:dual}.
    
    \paragraph{Why the general case is different.}
    Unlike the extreme cases $k=1$ and $k=n$, the general $k$-sum lateness problem admits no
    simple dispatching rule based only on individual job parameters, i.e., a fixed total order
    $\prec$ on the parameter space (e.g., on $\mathbb{R}^2$ for $(p_j,d_j)$) and a schedule that
    processes jobs in nondecreasing order of their own parameter vectors (unless
    such an order explicitly distinguishes between identical jobs).
    In particular, it cannot be optimized by interpolating between the \EDD{} and \SPT{} orders.
    For example, consider three jobs with $k=2$ and
    $(p_1,d_1)=(10,7)$, $(p_2,d_2)=(10,7)$, and $(p_3,d_3)=(9,8)$.
    Both the \EDD{} order $(1,2,3)$ and the \SPT{} order $(3,1,2)$ have objective value $34$,
    whereas the schedule $(1,3,2)$ has lateness values $3,11,22$ and thus $\lambda_2=33$.
    Thus, in general, no fixed per-job sorting rule can be optimal; in particular, one cannot obtain an
    optimum by interpolating between the \EDD{} and \SPT{} orders.

\paragraph{Our contributions.}
Let $D$ and $P$ denote the numbers of distinct due dates and distinct processing
times, respectively. We make four main contributions.

\begin{itemize}[leftmargin=2.2em]
\item \textbf{Complexity.} We give an objective-preserving reduction from
$k$-sum tardiness to $k$-sum lateness. Thus every algorithm for $k$-sum lateness
gives the corresponding algorithm for $k$-sum tardiness with only a linear-size
increase in the instance. Using this, we prove
that the decision version of $k$-sum lateness is weakly NP-complete when $k$ is
part of the input (Theorem~\ref{thm:NP-completeness}).

\item \textbf{Consequences from duality.} We use the standard dual representation of the sum
of the $k$ largest coordinates to reduce $k$-sum lateness to a family of
total-tardiness instances with uniformly shifted due dates. The dual identity
itself is standard; the algorithmic point is that one can restrict the shift to
a finite set of breakpoints and then use total-tardiness subroutines to solve
$k$-sum lateness. This gives a pseudopolynomial algorithm for integral data
(Theorem~\ref{thm:pseudopoly}), an additive approximation scheme with error at
most $\varepsilon M$ (Theorem~\ref{thm:additive-approx}), and XP algorithms for
fixed $P$ and fixed $D$ (Theorems~\ref{thm:fixed-P} and~\ref{thm:fixed-D-dual}).

\item \textbf{Structure and fixed $k$.} We prove that some optimal schedule is
$(k-1)$-adjacent to the \EDD{} order (Theorem~\ref{thm:edd-proximity}). This
gives an $O(n^{2k-1})$ XP algorithm by enumeration. We then prove a stronger
block--island decomposition (Theorem~\ref{thm:block-island}). In this
decomposition, the jobs outside a suitable top-$k$ set form due-date blocks,
and the top-$k$ jobs appear in islands between these blocks. The same theorem
also gives a suffix property inside each due-date class. These properties yield
an $O(k^2n^{k+2})$ algorithm for fixed $k$ (Theorem~\ref{thm:xp-k-improved}).
As a corollary, we resolve Woeginger's question on $k$-sum tardiness with
$c=2$ (Corollary~\ref{cor:woeginger-tardiness}).

\item \textbf{FPT algorithm.} The block--island decomposition and the suffix
property also yield an FPT algorithm parameterized by $D+k$, running in
$O(n\log n+k^2D^{k+2})$ time (Theorem~\ref{thm:fpt-kD-dp}).
\end{itemize}

Our results highlight an asymmetry between the small-$k$ and large-$k$ regimes. For fixed $k$, our algorithms rely on proximity to \EDD{}. For $k$ close to $n$, a similar proximity to \SPT{} does not hold, even when $n-k=1$ (Proposition~\ref{prop:no-spt-proximity-h1} in Section~\ref{sec:concl}).
    \paragraph{Organization.}
    Section~\ref{sec:prelim} fixes notation. Section~\ref{sec:hardness} proves weak NP-completeness
    when $k$ is part of the input (Theorem~\ref{thm:NP-completeness}).

    The paper then develops two largely independent algorithmic routes; readers may focus on
    either route without reading the other.
    \begin{itemize}[leftmargin=2.2em]
    \item \emph{Dual (shifted-total-tardiness) route:} Section~\ref{sec:dual} (in particular, Section~\ref{sec:dual-consequences}).
    This route yields a pseudopolynomial algorithm for integral data (Theorem~\ref{thm:pseudopoly}),
    a fully polynomial additive approximation scheme (Theorem~\ref{thm:additive-approx}),
    and XP algorithms for fixed $P$ and fixed $D$ (Theorems~\ref{thm:fixed-P} and~\ref{thm:fixed-D-dual}).

    \item \emph{Structural route:} Section~\ref{sec:structure} and its algorithmic consequences in Section~\ref{sec:xp-structural}.
    This route proves the block--island decomposition and the suffix property (Theorem~\ref{thm:block-island})
    and derives an $O(k^2n^{k+2})$ XP algorithm for fixed $k$ (Theorem~\ref{thm:xp-k-improved})
    and an FPT algorithm parameterized by $D+k$ (Theorem~\ref{thm:fpt-kD-dp}).
    \end{itemize}

\paragraph{Related work.}

For $k=1$ and $k=n$, our objective reduces to well-studied single-machine problems. Jackson's \EDD{} rule minimizes maximum lateness on one machine \cite{Jackson1955}, and Lawler extended this type of result to more general maximum-cost objectives with precedence constraints \cite{Lawler1973}. Smith's rule orders jobs by nondecreasing processing time and minimizes total completion time \cite{Smith1956}; since total lateness differs from total completion time by the constant $\sum_j d_j$, the same rule also minimizes total lateness. These extremes lie on the boundary of polynomially solvable single-machine due-date problems \cite{Graham1979,Brucker1977}.

The tardiness analogue is substantially harder. Du and Leung showed that total tardiness on one machine is weakly NP-hard \cite{DuLeung1990}. Lawler gave a pseudopolynomial dynamic program \cite{Lawler1977} and later an approximation scheme \cite{Lawler1982}. For the ordered objective that sums the $k$ largest tardiness values, Woeginger obtained an $O(n^{2k-1})$ algorithm for every fixed $k$ \cite{Woeginger1991}. At the end of his paper, he conjectured that the time complexity could be improved to $O(n^{k+c})$ for some small constant $c$. In this work, we resolve this conjecture by giving an $O(k^2 n^{k+2})$ algorithm (Corollary~\ref{cor:woeginger-tardiness}), corresponding to $c=2$.

Our fixed-$k$ algorithms are inspired by Woeginger's proximity approach for $k$-sum tardiness \cite{Woeginger1991}, where one searches for an optimum ``close'' to a canonical order. Extending this viewpoint from tardiness to lateness is not immediate, since lateness values can be negative and do not have the same truncation property as tardiness. These differences motivate the structural results developed in Section~\ref{sec:structure}.

Our objective also connects to the broader area of $k$-sum and ordered-median optimization. Gupta and Punnen introduced $k$-sum optimization as an interpolation between min-sum and min-max objectives \cite{GuptaPunnen1990}, and Punnen and Aneja extended the framework to more general feasible-set systems \cite{PunnenAneja1996}. Puerto, Rodriguez-Chia, and Tamir studied continuous and integer formulations for $k$-sum optimization \cite{PuertoRodriguezChiaTamir2017}. Related ordered weighted averaging (OWA) and ordered-median formulations were studied by Ogryczak and Tamir \cite{OgryczakTamir2003} and by Ogryczak and \'Sliwi\'nski \cite{OgryczakSliwinski2003}, and compact MILP formulations linking $k$-sum objectives with CVaR were developed by Filippi, Ogryczak, and Speranza \cite{FilippiOgryczakSperanza2019}. In scheduling, however, the ordered vector is induced by the schedule itself (lateness or tardiness values), which limits the direct applicability of generic $k$-sum formulations; see also \cite{Heeger2024}.

Recent work in parameterized scheduling studies due-date objectives with respect to the number of distinct due dates or processing times; e.g., Hermelin, Karhi, Pinedo, and Shabtay \cite{Hermelin2021} and Kaul, Mnich, and Molter \cite{DBLP:conf/iwpec/KaulMM24} consider minimizing the number of tardy jobs under such parameters. Our fixed-$D$ and fixed-$P$ algorithms extend this perspective to the $k$-sum lateness objective.

\section{Preliminaries}
    \label{sec:prelim}
    
    \paragraph{Instance and schedules.}
    For an integer $m\ge 1$, write $[m]=\{1,\ldots,m\}$.
    An instance consists of $n$ jobs $\Jobs=\{1,\dots,n\}$ with nonnegative
processing times $p_j\ge 0$, due dates $d_j\in\mathbb{R}$, and an integer
$k\in[n]$.
A schedule $S$ is a permutation of $\Jobs$ (with no idle time). Zero-processing
jobs are allowed and are ordered by the same tie-breaking rules as all other
jobs.

    Let $\pos_S(j)$ be the position of job $j$ in $S$.
    The completion time of $j$ is
    \[
    C_j(S)=\sum_{\ell:\,\pos_S(\ell)\le \pos_S(j)} p_\ell,
    \]
    and its lateness and tardiness are $L_j(S)=C_j(S)-d_j$ and $T_j(S)=\max\{0,L_j(S)\}$, respectively.

    \paragraph{Tie-breaking.}
    Fix a strict total order $\prec$ on jobs used only for tie-breaking.
    We refine processing time into a strict order $\prec_p$ by declaring $i\prec_p j$ if $p_i<p_j$, or if $p_i=p_j$ and $i\prec j$.
    Whenever we sort by \SPT{}, we mean sorting according to $\prec_p$.

\paragraph{Input model.}
Unless stated otherwise, we assume that processing times and due dates are rational numbers
encoded in binary, and that all arithmetic comparisons and additions are performed exactly.
When stating the pseudopolynomial algorithm and the additive approximation scheme, we
explicitly restrict the input to integral processing times and integral due dates. In contrast, our
structural statements remain valid regardless of the encoding model, as they rely only on
comparisons and algebraic identities.

    \paragraph{The top-$k$ objective.}
    Let $L_{(1)}(S)\ge\cdots\ge L_{(n)}(S)$ be the lateness values in
    non-increasing order. The $k$-sum lateness objective is
    \[
    \lambda_k(S)=\sum_{t=1}^k L_{(t)}(S).
    \]
    We use the convention $\max\emptyset=-\infty$. A set $X\subseteq\Jobs$ is a \emph{top-$k$ set for $S$} if $|X|=k$ and
    $\min_{x\in X} L_x(S)\ge \max_{y\notin X} L_y(S)$. For every top-$k$ set $X$
    for $S$, we have $\lambda_k(S)=\sum_{x\in X}L_x(S)$. When a canonical top-$k$
    set is needed, ties in lateness are broken by $\prec_p$, and the resulting
    set is denoted by $S[k]$. We also use the $k$-sum tardiness objective
    $\tau_k(S)=\sum_{t=1}^k \max\{0,L_{(t)}(S)\}$.

    \paragraph{Special cases and dispatching rules.}
    We make dispatching rules deterministic via the tie-breaking above: \EDD{} sorts by nondecreasing due dates and breaks ties by $\prec_p$, while \SPT{} sorts by $\prec_p$.
    For $k=1$, minimizing $\lambda_k$ is the maximum-lateness problem, and \EDD{} is optimal \cite{Jackson1955}.
    For $k=n$, minimizing $\lambda_k$ is equivalent to minimizing $\sum_j C_j(S)$, and \SPT{} is optimal \cite{Smith1956}.

\section{Weak NP-completeness of \texorpdfstring{$k$}{k}-sum lateness}
    \label{sec:hardness}

    This section proves that the decision version of $k$-sum lateness is weakly
NP-complete when $k$ is part of the input. We reduce from $k$-sum tardiness,
which is weakly NP-hard since it contains total tardiness as the special case
$k=n$ \cite{DuLeung1990}. We first give an objective-preserving reduction using
$k$ zero-processing dummy jobs. We then give a scaled version showing that the
hardness holds even under the stronger restriction that all processing times are
integers at least one.

    \begin{proposition}
    \label{prop:zero-dummy-reduction}
    There is an objective-preserving polynomial-time reduction from $k$-sum
tardiness to $k$-sum lateness using $k$ zero-processing dummy jobs.
    An instance with $n$ jobs is mapped to an instance with $n+k$ jobs, and any optimal schedule for the lateness instance can be transformed in linear time into an optimal schedule for the original tardiness instance.
    \end{proposition}
    
    \begin{proof}
    Let $\mathcal{I}$ be an instance of $k$-sum tardiness on job set $\Jobs$.
    Add a set $Z$ of $k$ dummy jobs, each with processing time $0$ and due date $0$, and keep all
    original jobs unchanged. The resulting instance is an instance of $k$-sum lateness on
    $\Jobs\cup Z$.

    First, consider a schedule $S$ of $\Jobs$. Let $Z[S]$ be the schedule that places all jobs
    of $Z$ first and then schedules $\Jobs$ according to $S$. Since the dummy jobs do not take
    time, the lateness of each original job is unchanged, and each dummy job has lateness $0$.
    Because there are $k$ such zero-lateness jobs, the $k$ largest lateness values in $Z[S]$
    are exactly the positive part of the $k$ largest lateness values of $S$. Hence
    \[
    \lambda_k(Z[S])=\tau_k(S).
    \]
    
    Now let $Q$ be any schedule of $\Jobs\cup Z$, and let $S=Q|_{\Jobs}$ be the restriction of
    $Q$ to the original jobs. Moving all dummy jobs to the beginning does not change the
    completion time of any original job, since dummy jobs have processing time $0$. It changes
    each dummy lateness to $0$, and therefore cannot increase the top-$k$ lateness sum. Thus
    \[
    \lambda_k(Q)\ge \lambda_k(Z[S])=\tau_k(S).
    \]
    
    Let $\OPT_L$ be the optimal value of the constructed lateness instance, and let $\OPT_T$
    be the optimal value of the original tardiness instance. The first part gives
    $\OPT_L\le \OPT_T$, by applying it to an optimal tardiness schedule. The second part gives
    $\OPT_L\ge \OPT_T$, because every lateness schedule has value at least the tardiness value
    of its restriction, and this restriction has value at least $\OPT_T$. Therefore
    \[
    \OPT_L=\OPT_T.
    \]
    
    Finally, let $Q$ be an optimal lateness schedule and set $S=Q|_{\Jobs}$. By the inequality
    above,
    \[
    \OPT_T=\OPT_L=\lambda_k(Q)\ge \tau_k(S)\ge \OPT_T.
    \]
    Hence $\tau_k(S)=\OPT_T$, so the restriction of any optimal lateness schedule is an
    optimal tardiness schedule for $\mathcal{I}$.
    \end{proof}
    
    Since zero-processing jobs are part of our model, Proposition~\ref{prop:zero-dummy-reduction}
    already implies weak NP-hardness. The next theorem records the stronger fact that the
    hardness persists even when all processing times are at least one.

    \begin{theorem}\label{thm:NP-completeness}
    The decision version of $k$-sum lateness is NP-complete when $k$ is part of the input, even
    when all processing times are integers at least one, all due dates are integers, and the
    decision threshold is an integer. Moreover, together with the pseudopolynomial algorithm of
    Theorem~\ref{thm:pseudopoly}, this implies weak NP-completeness.
    \end{theorem}

    \begin{proof}
    Membership in NP is immediate: a certificate is a permutation of the jobs, and its objective
    value can be computed in polynomial time. We reduce from the decision version of total
    tardiness, which is weakly NP-hard by Du and Leung~\cite{DuLeung1990}.

    Let $\mathcal{I}$ be a total-tardiness instance on job set $\Jobs$, with $n=|\Jobs|$, positive
    integral processing times, integral due dates, and an integral target value $B$. Set $k:=n$
    and view $\mathcal{I}$ as a $k$-sum tardiness instance. Let $\Delta_k=k(k+1)/2$ and choose an
    integer $M>\Delta_k+k^2$. Construct a lateness instance as follows.
    For every original job $j$, set $p'_j=Mp_j$ and $d'_j=Md_j+k$.
    Add $k$ dummy jobs, each with processing time $1$ and due date $0$.
    Set the lateness target to $B'=MB+\Delta_k$.
    
    If $S$ is a schedule of the original jobs with $\tau_k(S)\le B$, let $X_S$ be the schedule
    that places the dummy jobs first and then schedules the original jobs according to $S$. In
    $X_S$, the dummy jobs have lateness values $1,2,\ldots,k$. For every original job $j$, the
    $k$ units of dummy processing time are canceled by the $+k$ shift in the due date, and thus
    \[
    L'_j(X_S)=M\bigl(C_j(S)-d_j\bigr)=M L_j(S).
    \]
    In particular, original jobs with positive lateness in $S$ have constructed lateness at least
    $M>k$, while original jobs with nonpositive lateness in $S$ have constructed lateness at most
    $0$. Hence the $k$ largest lateness values of $X_S$ come from the positive lateness values of
    $S$ (scaled by $M$) and, if fewer than $k$ such values exist, from some of the dummy values
    $1,\ldots,k$. Therefore
    \[
    \lambda_k(X_S) \le M\tau_k(S)+\Delta_k \le MB+\Delta_k = B'.
    \]
    
    Conversely, let $Q$ be any schedule of the constructed lateness instance, and let
    $S=Q|_{\Jobs}$. For an original job $j$, let $q_j$ be the number of dummy jobs scheduled
    before $j$ in $Q$. Then
    \[
    L'_j(Q)=MC_j(S)+q_j-(Md_j+k)=M\bigl(C_j(S)-d_j\bigr)-(k-q_j)\ge M\bigl(C_j(S)-d_j\bigr)-k.
    \]
    Let $H$ be the set of original jobs with positive tardiness in $S$ that contribute to
    $\tau_k(S)$ (if fewer than $k$ jobs have positive tardiness, take all of them). Every job in
    $H$ has $C_j(S)-d_j\ge 1$ and thus $L'_j(Q)\ge M-k>0$.

    Since the $k$ dummy jobs also have positive lateness, the $k$ largest lateness values in $Q$
    are all positive, and therefore
    \[
    \lambda_k(Q)\ge \sum_{j\in H} L'_j(Q)\ge \sum_{j\in H} \bigl(M(C_j(S)-d_j)-k\bigr)=M\tau_k(S)-k|H|\ge M\tau_k(S)-k^2.
    \]
    If $\tau_k(S)\ge B+1$, then $\lambda_k(Q)\ge M(B+1)-k^2>MB+\Delta_k=B'$ by the choice of
    $M$. Therefore, any schedule with $\lambda_k(Q)\le B'$ yields a schedule $S$ with
    $\tau_k(S)\le B$.
    The construction has polynomial encoding length, and the argument proves
    NP-hardness. Membership in NP was shown above, so the decision problem is
    NP-complete. The pseudopolynomial algorithm given later in
    Theorem~\ref{thm:pseudopoly} shows that this NP-completeness is weak.
    \end{proof}

\section{From top-\texorpdfstring{$k$}{k} lateness to shifted total tardiness using duality}
    \label{sec:dual}

    This section presents an algorithmic route based on a continuous shift parameter $r$.
    Using a standard dual representation of the sum of the $k$ largest lateness values, we reduce
    $k$-sum lateness to a family of shifted total-tardiness instances with due dates $d_j+r$.
    We then show how to restrict attention to a finite set of candidate shifts and how to solve
    each shifted instance using variants of Lawler's dynamic program.

    This yields a pseudopolynomial algorithm and a fully polynomial additive approximation scheme
    for integral data, as well as XP algorithms for fixed $P$ and fixed $D$.
    The section is independent of the combinatorial structural results developed later.
    We first derive the min--min formulation, then recall the total-tardiness subroutines used to compute $\Phi(r)$, and finally restrict the relevant shifts for the resulting algorithms.

\paragraph{Duality.}
For an arbitrary schedule $S$, the top-$k$ objective can be written as the following 0--1 linear program:
    \begin{equation}
    \label{eq:topk-as-max}
    \lambda_k(S)=
    \max\Bigl\{\sum_{j=1}^n x_j L_j(S): \sum_{j=1}^n x_j=k,\; x\in\{0,1\}^n\Bigr\}.
    \end{equation}

    Rewriting $L_j(S)=C_j(S)-d_j$ and relaxing $x\in\{0,1\}^n$ to $x\in[0,1]^n$ gives
    \[
    \max\left\{\sum_{j=1}^n x_j(C_j(S)-d_j):\sum_{j=1}^n x_j=k,\ 0\le x_j\le 1\ \forall j\right\}.
    \]
    This relaxation has the same optimal value as~\eqref{eq:topk-as-max}, since the polytope
    $\{x\in[0,1]^n:\sum_j x_j=k\}$ is integral (e.g., its constraint matrix is totally unimodular and the
    right-hand side is integral). Its dual has a free variable $r$ for the equality constraint and
    nonnegative variables $y_j$ for the upper bounds $x_j\le 1$:
    \begin{align}
    \lambda_k(S)
      &= \min\Bigl\{kr+\sum_{j=1}^n y_j : r+y_j\ge C_j(S)-d_j,\; y\ge 0,\; r\in\mathbb{R}\Bigr\}\notag \\
      &= \min_{r\in\mathbb{R}} \Bigl\{kr+\sum_{j=1}^n (C_j(S)-d_j-r)_+\Bigr\}.\label{eq:dual}
    \end{align}

    Therefore, by minimizing over all possible schedules,
    \begin{equation}
    \label{eq:minmin}
    \min_S \lambda_k(S)=
    \min_{r\in\mathbb{R}} \Bigl\{kr+ \min_S \sum_{j=1}^n (C_j(S)-(d_j+r))_+ \Bigr\}.
    \end{equation}

    \begin{theorem}[Top-$k$ dual reduction]
    \label{thm:topk-dual-reduction}
    For every instance, the identity~\eqref{eq:minmin} holds. Thus, after fixing $r$, the
    inner problem is a total-tardiness instance with due dates $d_j+r$.
    \end{theorem}

    \begin{proof}
    The derivation above proves the dual identity~\eqref{eq:dual} for every
    fixed schedule $S$. Minimizing that identity over all schedules gives
    \eqref{eq:minmin}. For a fixed value of $r$, the inner objective is
    $\sum_{j=1}^n (C_j(S)-(d_j+r))_+$, which is exactly total tardiness for the
    instance with shifted due dates $d_j+r$. This proves the claim.
    \end{proof}

    The inner minimization is total tardiness for due dates shifted by $r$. Thus, once a shift
    $r$ is fixed, the remaining task is a total-tardiness instance. For integral data we use
    Lawler's pseudopolynomial algorithm for total tardiness.

The dual formulation reduces $k$-sum lateness to a family of shifted
total-tardiness instances. For a real number $r$, let $\Phi(r)$ be the
optimum value of the total-tardiness instance obtained by shifting every due
date from $d_j$ to $d_j+r$:
\[
\Phi(r)
=
\min_S \sum_{j\in\Jobs} (C_j(S)-d_j-r)_+ .
\]
Thus the dual formulation~\eqref{eq:minmin} can be written as
\[
\OPT
=
\min_{r\in\mathbb{R}}
\{kr+\Phi(r)\}.
\tag{Shift-Total Tardiness}
\label{eq:shift-tt}
\]
The algorithmic consequences of this shifted formulation depend on two tasks:
restricting the set of shifts $r$ that must be considered, and computing $\Phi(r)$
efficiently for each such shift.

\subsection{Revisiting total tardiness}
\label{sec:tt-subroutines}

We next recall Lawler's decomposition dynamic program for single-machine total tardiness,
which serves as our main subroutine. We then describe two refinements tailored to our use
in the shifted formulation: a compressed-state implementation for a constant number $P$ of
distinct processing times, and a direct dynamic program for a constant number $D$ of
distinct due dates.
\subsubsection{Lawler's decomposition dynamic program}
\label{sec:lawler-dp}

We recall Lawler's dynamic program for single-machine total tardiness, which is the main subroutine used below.

Fix an arbitrary total-tardiness instance with due dates $\bar d_j$. Let
\[
E=(e_1,e_2,\ldots,e_n)
\]
be the order of the jobs by nondecreasing due date, using the fixed tie-breaking
order from Section~\ref{sec:prelim}. For a job $j$, let $\pos_E(j)$ be its
position in this order. In this subsection, an interval $[a,b]$ always refers
to positions in the order $E$.

Let $\prec_p$ be the strict processing-time order from
Section~\ref{sec:prelim}. For a threshold job $q$, define
\[
J(a,b,q)
=
\{h\in\Jobs:
a\le \pos_E(h)\le b,\ h\prec_p q\}.
\]
The threshold job $q$ is used only to restrict the processing-time order. It
need not belong to the interval $[a,b]$. We also use a dummy threshold job
$q_0$ satisfying $h\prec_p q_0$ for every real job $h$.

Let $F(a,b,q,t)$ be the minimum total tardiness, with due dates $\bar d_j$, of the jobs in
$J(a,b,q)$ when the subproblem starts at time $t$. The boundary conditions are
\[
F(a,b,q,t)=0
\quad\text{if }J(a,b,q)=\emptyset,
\]
and
\[
F(a,b,q,t)=(t+p_j-\bar d_j)_+
\quad\text{if }J(a,b,q)=\{j\}.
\]

Assume now that $|J(a,b,q)|\ge 2$. Let $u$ be the $\prec_p$-maximum job in $J(a,b,q)$, and set
$s:=\pos_E(u)$. For each integer $\alpha$ with $0\le \alpha\le b-s$, set
$A_\alpha:=J(a,s+\alpha,u)$. In Lawler's decomposition, the jobs in $A_\alpha$
are scheduled before $u$, and the jobs in $J(s+\alpha+1,b,u)$ are scheduled
after $u$.
For such a value of $\alpha$, define
\[
C_u(\alpha)
=
t+p_u+\sum_{h\in J(a,s+\alpha,u)}p_h .
\]
Then Lawler shows the following recurrence
\[
F(a,b,q,t)
=
\min_{0\le \alpha\le b-s}
\left\{
F(a,s+\alpha,u,t)
+
(C_u(\alpha)-\bar d_u)_+
+
F(s+\alpha+1,b,u,C_u(\alpha))
\right\}.
\tag{Lawler-DP}
\label{eq:lawler-dp}
\]
The value of the total-tardiness instance is $
F(1,n,q_0,0).$

\begin{theorem}[Lawler]
\label{thm:lawler-pseudopoly}
For integral processing times, total tardiness on one machine can be solved by
recurrence~\eqref{eq:lawler-dp} in time
\[
O(n^4(P_\Sigma+1)),
\qquad
P_\Sigma=\sum_{j\in\Jobs}p_j .
\]
That is, there is a pseudopolynomial-time algorithm for total tardiness.
\end{theorem}

\begin{proof}
We only account for the running time of recurrence~\eqref{eq:lawler-dp}; its correctness is due to
Lawler~\cite{Lawler1977}.

Since processing times are integral, every start time $t$ that appears in a subproblem is an
integer in $\{0,\ldots,P_\Sigma\}$. There are $O(n^2)$ choices for the interval $[a,b]$, $O(n)$
choices for the threshold job $q$, and $O(P_\Sigma+1)$ possible values of $t$, so the number of
states $(a,b,q,t)$ is $O(n^3(P_\Sigma+1))$.

For each state, the recurrence considers at most $n$ values of the split parameter $\alpha$.
After a standard preprocessing step that stores cumulative sums of processing times along the
relevant job orders (for each possible threshold job $u$), the term
$\sum_{h\in J(a,s+\alpha,u)} p_h$ can be obtained in $O(1)$ time. Hence each transition is
evaluated in $O(1)$ time after preprocessing.

The total running time is therefore
$O(n^3(P_\Sigma+1)\cdot n)=O(n^4(P_\Sigma+1))$.
\end{proof}

\subsubsection{A compressed implementation for fixed \texorpdfstring{$P$}{P}}
\label{sec:lawler-fixed-P}

We also need Lawler's dynamic program when the number of distinct processing
times is fixed.  Let
\[
\rho_1,\ldots,\rho_P
\]
be the distinct processing times. One of these values may be $0$.

In Lawler's pseudopolynomial implementation, the state stores the start time
$t$.  When $P$ is fixed, we do not need to store $t$ as an arbitrary integer.
Instead, we store how many previously scheduled jobs have each processing time.

For a vector $c=(c_1,\ldots,c_P)\in\{0,\ldots,n\}^P$, define $t(c)=\sum_{\ell=1}^P c_\ell\rho_\ell$.
Thus $c_\ell$ is the number of jobs of processing time $\rho_\ell$ that have
already been scheduled, and $t(c)$ is the corresponding start time.

The vector $c$ is not meant to encode the set of jobs already scheduled. It
only encodes the start time of the subproblem. This is sufficient because
Lawler's recurrence depends on the previously scheduled jobs only through the
start time.

For rational input data, all values $t(c)$ and all tardiness terms are rational
numbers of polynomial bit length for fixed $P$, and they can be compared and
added exactly. Thus, by counting the number of DP states below, we obtain a
polynomial-time algorithm for fixed $P$.

For a set $A\subseteq\Jobs$, define
\[
\chi(A)=(\chi_1(A),\ldots,\chi_P(A)),
\]
where $\chi_\ell(A)$ is the number of jobs in $A$ with processing time
$\rho_\ell$.  If $u$ is a job with processing time $\rho_\ell$, let $e(u)$ be
the vector with a $1$ in coordinate $\ell$ and $0$ in all other coordinates.

The compressed state is $H(a,b,q,c)$. It has the same meaning as Lawler's state
$F(a,b,q,t)$, except that the start time is $t(c)$ instead of an explicitly
stored integer $t$. The boundary conditions are
\[
H(a,b,q,c)=0
\quad\text{if }J(a,b,q)=\emptyset,
\]
and
\[
H(a,b,q,c)=(t(c)+p_j-\bar d_j)_+
\quad\text{if }J(a,b,q)=\{j\}.
\]

Assume now that $|J(a,b,q)|\ge 2$. Let $u$ be the $\prec_p$-maximum job in
$J(a,b,q)$, and set $s:=\pos_E(u)$. For each integer $\alpha$ with
$0\le \alpha\le b-s$, set $A_\alpha:=J(a,s+\alpha,u)$. In the corresponding
Lawler decomposition, the jobs in $A_\alpha$ are scheduled before $u$, and the
jobs in $J(s+\alpha+1,b,u)$ are scheduled after $u$.

If the subproblem starts at time $t(c)$, then $u$ completes at time
$t(c)+\sum_{h\in A_\alpha}p_h+p_u$. After scheduling $A_\alpha$ and then $u$,
the vector $c$ becomes $c+\chi(A_\alpha)+e(u)$.
Therefore the compressed recurrence is
\[
\begin{aligned}
H(a,b,q,c)
=
\min_{0\le \alpha\le b-s}
\Big\{&
H(a,s+\alpha,u,c)
+
\left(
t(c)+\sum_{h\in A_\alpha}p_h+p_u-\bar d_u
\right)_+
\\
&+
H\!\Bigl(
s+\alpha+1,b,u,
c+\chi(A_\alpha)+e(u)
\Bigr)
\Big\}.
\end{aligned}
\tag{Lawler-DP-P}
\label{eq:lawler-dp-fixed-P}
\]
The value of the total-tardiness instance is $H(1,n,q_0,0)$, where $0$ denotes
the all-zero vector.

We allow all vectors $c\in\{0,\ldots,n\}^P$ in the state space. States that are
not reachable from the initial state may be computed, but they do not affect the
value of $H(1,n,q_0,0)$. If a transition produces a vector outside this box, we
ignore that transition.

\begin{theorem}
\label{thm:TT-fixed-P}
If there are $P$ distinct processing times, then total tardiness on one machine
can be solved in time $O(Pn^{P+4})$. In particular, for fixed $P$, this is
$O(n^{P+4})$.
\end{theorem}
\begin{proof}
We use the same recurrence as Lawler, but we replace the scalar start time by
the vector $c$. The recurrence remains valid when some processing times are
zero; the fixed order $\prec_p$ makes the maximum job in each set $J(a,b,q)$
unique. The intended invariant is that $H(a,b,q,c)$ has the same value as
Lawler's state $F(a,b,q,t(c))$, where $t(c)=\sum_{\ell=1}^P c_\ell\rho_\ell$.
Each transition schedules the same block as in Lawler's recurrence and updates
$c$ by adding the processing-time count vector of that block and of the pivot
job. Hence the compressed recurrence is exactly Lawler's recurrence evaluated at
the represented start time.
The recurrence is over states $(a,b,q,c)$, where $c$ is the processing-time
count vector that represents the start time. Consider a state with
$|J(a,b,q)|\ge 2$, and let $u$ be the $\prec_p$-maximum job in $J(a,b,q)$.
For each integer $\alpha$ with $0\le \alpha\le b-s$, the two recursive calls are
\[
H(a,s+\alpha,u,c)
\quad\text{and}\quad
H(s+\alpha+1,b,u,c+\chi(A_\alpha)+e(u)).
\]
Since the third argument in both calls is $u$, the job set of each recursive
subproblem is defined using the condition $h\prec_p u$. Thus neither recursive
subproblem contains $u$. Consequently, the job set strictly decreases in every
recursive call. The vector $c$ may change in the second call, but it only
records the new start time; it does not change which jobs belong to the
subproblem. Therefore the recurrence has no cyclic dependence.

There are $O(n^2)$ choices of the interval $[a,b]$ and $O(n)$ choices of the
threshold job $q$.  Thus there are $O(n^3)$ choices of $(a,b,q)$.  Since
$c\in\{0,\ldots,n\}^P$, there are $O(n^P)$ possible vectors $c$.  Hence the
number of states $(a,b,q,c)$ is $O(n^{P+3})$.

For each state, the recurrence tries at most $n$ values of $\alpha$. The values
$\sum_{h\in A_\alpha}p_h$ and $\chi(A_\alpha)$ can be precomputed, or computed
from prefix counts for each possible threshold job. Updating and indexing the
vector
\[
c+\chi(A_\alpha)+e(u)
\]
takes $O(P)$ time in a uniform implementation.

Thus each state has $O(n)$ transitions, and each transition can be evaluated in
$O(P)$ time.  The total running time is
\[
O(n^{P+3}\cdot n\cdot P)=O(Pn^{P+4}).
\]
For fixed $P$, this is $O(n^{P+4})$.\qedhere

\end{proof}

\subsubsection{A dynamic program for fixed \texorpdfstring{$D$}{D}}
\label{sec:TT-fixed-D}

We now consider the shifted total-tardiness problem when the number of distinct
due dates is fixed. Total tardiness with a fixed number of due dates is known
to be polynomially solvable; see Tian, Ng, and Cheng~\cite{TianNC05}. We use a
direct dynamic program, since this is the form needed later for computing
$\Phi(r)$.

Let
\[
\delta_1<\cdots<\delta_D
\]
be the distinct due dates. For each $i\in[D]$, let
\[
J_i=\{j\in\Jobs:d_j=\delta_i\},
\qquad
n_i=|J_i|.
\]
Sort the jobs in $J_i$ according to the strict SPT order $\prec_p$. From now on
in this subsection, we regard $J_i$ as this ordered list, and write $J_i[h]$ for
its $h$-th job.

\begin{lemma}
\label{lem:spt-within-due-date}
Fix a shift $r$. In the total-tardiness instance with due dates $d_j+r$, there
is an optimal schedule in which, for every $i\in[D]$, the jobs of $J_i$ appear
in the order
\[
J_i[1],J_i[2],\ldots,J_i[n_i].
\]
\end{lemma}

\begin{proof}
Start with an optimal schedule $S$. Suppose that, for some $i\in[D]$, two jobs
$a,b\in J_i$ appear in the wrong order: $a$ is scheduled before $b$, but
$b\prec_p a$. Then
\[
d_a=d_b=\delta_i
\qquad\text{and}\qquad
p_b\le p_a.
\]
The jobs $a$ and $b$ need not be consecutive.

Let $H$ be the block of jobs scheduled strictly between $a$ and $b$. Form a new
schedule $S'$ by exchanging the positions of $a$ and $b$, while keeping the
jobs of $H$ in the same relative order.

Compare completion times. Job $b$ in $S'$ completes no later than job $a$ did
in $S$, because $p_b\le p_a$. Job $a$ in $S'$ completes at exactly the same time
at which job $b$ completed in $S$, because the total processing time of
$\{a\}\cup H\cup \{b\}$ is unchanged. Every job in $H$ completes no later in $S'$
than in $S$, since job $b$ (with $p_b\le p_a$) is moved before $H$. All other
jobs have the same completion time in both schedules.

Since $(\cdot)^+$ is nondecreasing in the completion time $C$, and $a$ and $b$ have the same
shifted due date $\delta_i+r$, the total shifted tardiness of $a$ and $b$ does
not increase. The shifted tardiness of every job in $H$ also does not increase,
and all other contributions are unchanged. Hence $S'$ is also optimal.

By repeating this process, we eliminate all inversions inside each list $J_i$ and obtain an optimal schedule in which every $J_i$ is sorted.
\end{proof}

By Lemma~\ref{lem:spt-within-due-date}, it is enough to consider schedules
obtained by merging the $D$ ordered lists $J_1,\ldots,J_D$. Equivalently, a
schedule is built by repeatedly choosing one list and taking the next
unscheduled job from that list. For $h=0,\ldots,n_i$, let
$Q_i(h)=\sum_{\ell=1}^h p_{J_i[\ell]}$, with $Q_i(0)=0$. Thus $Q_i(h)$ is the
total processing time of the first $h$ jobs of list $J_i$.

For a vector $a=(a_1,\ldots,a_D)$ with $0\le a_i\le n_i$, define
$Q(a)=\sum_{i=1}^D Q_i(a_i)$. Thus $Q(a)$ is the total processing time of the
partial schedule containing the first $a_i$ jobs of each list $J_i$.

Fix a shift $r$. Let $G_r(a)$ be the minimum shifted total tardiness among all
merges of the first $a_i$ jobs of each list $J_i$. The boundary condition is
$G_r(0,\ldots,0)=0$.

Now consider a nonzero vector $a$. If the last job is taken from list $J_i$,
then $a_i>0$, and the previous vector is $a-e_i$. The last job has completion
time $Q(a)$ and shifted due date $\delta_i+r$. Therefore
\[
G_r(a)
=
\min_{i:a_i>0}
\left\{
G_r(a-e_i)+(Q(a)-\delta_i-r)_+
\right\},
\tag{DP-D}
\label{eq:DP-D}
\]
where $e_i$ is the $i$-th unit vector.

The value of the shifted total-tardiness instance is $\Phi(r)=G_r(n_1,\ldots,n_D)$.

\begin{lemma}
\label{lem:TT-fixed-D}
After the lists $J_1,\ldots,J_D$ have been sorted, a shifted total-tardiness
instance with $D$ distinct due dates can be solved in time
$O\left(D\prod_{i=1}^D(n_i+1)\right)\le O(D(n+1)^D)$. 
\end{lemma}

\begin{proof}
Lemma~\ref{lem:spt-within-due-date} shows that it is enough to optimize over
merges of the ordered lists $J_1,\ldots,J_D$.

For a vector $a$, every merge of the first $a_i$ jobs of each list has the same
total processing time, namely $Q(a)$. We prove by induction on $a_1+\cdots+a_D$
that $G_r(a)$ is the minimum shifted total tardiness among all such merges.

The claim is clear for $a=(0,\ldots,0)$. Let $a\ne 0$. In any merge counted by
$a$, the last job must come from some list $J_i$ with $a_i>0$. If this last job
comes from $J_i$, then the previous vector is $a-e_i$. The completion time of
the last job is $Q(a)$, and its shifted due date is $\delta_i+r$. Hence its
shifted tardiness is $(Q(a)-\delta_i-r)_+$.

By the induction hypothesis, the best value for the previous part is
$G_r(a-e_i)$. Minimizing over all possible choices of the last list gives
recurrence~\eqref{eq:DP-D}.

There are $\prod_{i=1}^D(n_i+1)$ vectors $a$. For each vector, the recurrence
checks at most $D$ choices of $i$. The values $Q_i(h)$ are precomputed. We store
the DP table in a mixed-radix array indexed by the vector $a=(a_1,\ldots,a_D)$;
within the same asymptotic bound we compute all values $Q(a)$ and the array
indices of all valid predecessors $a-e_i$. Hence each transition is evaluated
in constant time after this preprocessing.
Therefore the running time is $O\left(D\prod_{i=1}^D(n_i+1)\right)$. Since
$n_i\le n$ for every $i$, this is at most $O(D(n+1)^D)$.
\end{proof}

\subsection{Algorithmic consequences}
\label{sec:dual-consequences}

Recall from Section~\ref{sec:tt-subroutines} that, for a real number $r$,
\[
\Phi(r)
=
\min_S \sum_{j\in\Jobs}(C_j(S)-d_j-r)_+,
\]
and that the dual formulation gives
\[
\OPT
=
\min_{r\in\mathbb{R}}\{kr+\Phi(r)\}.
\]
This is exactly the identity in~\eqref{eq:shift-tt}. We use it to obtain four algorithmic
consequences of the shifted formulation for $k$-sum lateness. The total-tardiness
subroutines for computing $\Phi(r)$ were given in Section~\ref{sec:tt-subroutines}; here we
only specify which values of $r$ must be tried.

\paragraph{Candidate shifts (breakpoints).}

For a fixed schedule $S$, define
\[
f_S(r)
=
kr+\sum_{j\in\Jobs}(C_j(S)-d_j-r)_+ .
\]
This is the value obtained by using schedule $S$ in the shifted
total-tardiness formulation.

\begin{lemma}
\label{lem:breakpoint-reduction}
The minimum in~\eqref{eq:shift-tt} is attained at a value
\[
r=C_j(S)-d_j
\]
for some job $j$ and some schedule $S$.
\end{lemma}

\begin{proof}
Such a minimizing pair exists. Indeed, there are finitely many schedules. For
each fixed schedule $S$, the function $f_S(r)$ is piecewise linear, has slope
$k>0$ for all sufficiently large $r$, and has slope $k-n\le 0$ for all
sufficiently small $r$. If $k<n$, the function tends to $+\infty$ at both
ends. If $k=n$, it is constant on the leftmost unbounded interval, and the same
value is attained at the leftmost breakpoint. Hence the minimum over all
schedules and all real $r$ is attained.

Choose a pair $(S^\star,r^\star)$ minimizing $f_S(r)$ over all schedules $S$
and all real numbers $r$. This is the same as minimizing the right-hand side
of~\eqref{eq:shift-tt}.

Fix the schedule $S^\star$. The function $f_{S^\star}$ is piecewise linear,
and its breakpoints are among the values $C_j(S^\star)-d_j$ for $j\in\Jobs$.
On an interval with no breakpoint, the function is affine. If the slope is
nonzero, an interior point of the interval is not needed for the minimum. If
the slope is zero, the same value is attained at an endpoint. The unbounded
intervals cause no problem. For large $r$ the slope is $k>0$. For small $r$
the slope is $k-n\le 0$; if $k=n$, the function is constant on the leftmost
unbounded interval, so the same value is attained at the leftmost breakpoint.

Therefore there is a breakpoint $r'$ of $f_{S^\star}$ with
$f_{S^\star}(r')=f_{S^\star}(r^\star)$. Since
$kr'+\Phi(r')\le f_{S^\star}(r')$,
and $(S^\star,r^\star)$ was globally optimal, equality must hold. Hence the
minimum in~\eqref{eq:shift-tt} is attained at this breakpoint.
\end{proof}

\subsubsection{A pseudopolynomial algorithm for \texorpdfstring{$k$}{k}-sum lateness}
\label{sec:pseudo}

Assume in this subsection that all processing times and due dates are integral.
Let $P_\Sigma=\sum_{j\in\Jobs}p_j$. For every schedule, every completion time
belongs to $\{0,\ldots,P_\Sigma\}$. Hence every breakpoint from
Lemma~\ref{lem:breakpoint-reduction} belongs to
\[
R=\{t-d_j:j\in\Jobs,\ t\in\{0,\ldots,P_\Sigma\}\}.
\]
Thus it is enough to try the shifts in $R$.

\begin{algorithm}[htbp]
\caption{Pseudopolynomial algorithm for integral $k$-sum lateness}
\label{alg:pseudopoly}
\DontPrintSemicolon
$P_\Sigma\leftarrow\sum_{j\in\Jobs}p_j$\;
$R\leftarrow\{t-d_j:j\in\Jobs,\ t\in\{0,\ldots,P_\Sigma\}\}$\;
$v_{\mathrm{best}}\leftarrow+\infty$; $S_{\mathrm{best}}\leftarrow\textsc{null}$\;
\ForEach{$r\in R$}{
  Solve the shifted total-tardiness instance with due dates $d_j+r$ using Theorem~\ref{thm:lawler-pseudopoly}\;
  Let $S_r$ be the returned schedule, and let $\Phi(r)$ be its value\;
  \If{$kr+\Phi(r)<v_{\mathrm{best}}$}{
    $v_{\mathrm{best}}\leftarrow kr+\Phi(r)$\;
    $S_{\mathrm{best}}\leftarrow S_r$\;
  }
}
\Return $S_{\mathrm{best}}$\;
\end{algorithm}

\begin{theorem}
\label{thm:pseudopoly}
For integral data, $k$-sum lateness admits a pseudopolynomial algorithm. A
direct implementation of Algorithm~\ref{alg:pseudopoly} runs in time
\[
O(n^5(P_\Sigma+1)^2).
\]
\end{theorem}

\begin{proof}
By Lemma~\ref{lem:breakpoint-reduction} and the definition of $R$, the minimum
in~\eqref{eq:shift-tt} is attained at some shift $r\in R$. The set $R$ has size
at most $n(P_\Sigma+1)$.

For each $r\in R$, Theorem~\ref{thm:lawler-pseudopoly} computes $\Phi(r)$ in
time $O(n^4(P_\Sigma+1))$. The algorithm returns the schedule that minimizes
$kr+\Phi(r)$ over all $r\in R$. Hence it is optimal by~\eqref{eq:shift-tt}.
The running time is $O(n(P_\Sigma+1)\cdot n^4(P_\Sigma+1))=O(n^5(P_\Sigma+1)^2)$.\qedhere
\end{proof}

\subsubsection{A fully polynomial additive approximation scheme}
\label{sec:additive-approx}

Since lateness values may be negative, a multiplicative approximation guarantee
is not meaningful in general, so we use an additive one. For integral data, by
a fully polynomial additive approximation scheme we mean an algorithm that, for
every $\varepsilon\in(0,1]$, runs in time polynomial in $n$, $k$, and
$1/\varepsilon$, and returns a schedule $S$ such that
\[
\lambda_k(S)\le \OPT+\varepsilon M,
\qquad
M=\max\{1,\max_j p_j,\max_j |d_j|\}.
\]

Fix $\varepsilon\in(0,1]$ and set
\[
\Delta=\frac{\varepsilon M}{4k(n+1)}.
\]
For each job $j$, choose multiples $\tilde p_j$ and $\tilde d_j$ of $\Delta$
with $\tilde p_j\ge 0$ such that
\[
|p_j-\tilde p_j|\le \Delta,
\qquad
|d_j-\tilde d_j|\le \Delta.
\]
For example, one may round $p_j$ down to the nearest nonnegative multiple of
$\Delta$ and round $d_j$ to the nearest multiple of $\Delta$. Scale the rounded
instance by $1/\Delta$. The scaled processing times and due dates are integers.
Moreover,
\[
\sum_j \frac{\tilde p_j}{\Delta}
\le
\frac{\sum_j p_j}{\Delta}+n
=
O\left(\frac{kn^2}{\varepsilon}\right).
\]
We solve this scaled instance exactly using Theorem~\ref{thm:pseudopoly} and
return the resulting schedule for the original instance. Let $\tilde\lambda_k$
denote the objective value in the rounded instance.

\begin{lemma}
\label{lem:rounding-error}
For every schedule $S$,
\[
|\lambda_k(S)-\tilde\lambda_k(S)|
\le
k(n+1)\Delta.
\]
\end{lemma}

\begin{proof}
Fix a schedule $S$. Each completion time changes by at most $n\Delta$, because
it is the sum of at most $n$ processing times. Each due date changes by at most
$\Delta$. Hence each lateness value changes by at most $(n+1)\Delta$.

If two vectors differ coordinatewise by at most $\eta$, then the sums of their
$k$ largest entries differ by at most $k\eta$. Applying this with
$\eta=(n+1)\Delta$ gives the claim.
\end{proof}

\begin{theorem}
\label{thm:additive-approx}
For integral data, $k$-sum lateness admits a fully polynomial additive
approximation scheme. Given $\varepsilon\in(0,1]$, it returns a schedule of
value at most
\[
\OPT+\varepsilon M,
\qquad
M=\max\{1,\max_j p_j,\max_j |d_j|\}.
\]
Its running time is polynomial in $n$, $k$, and $1/\varepsilon$.
\end{theorem}

\begin{proof}
Let $S^\star$ be optimal for the original instance. Let $\widetilde S$ be
optimal for the rounded instance. By Lemma~\ref{lem:rounding-error},
\[
\lambda_k(\widetilde S)
\le
\tilde\lambda_k(\widetilde S)+k(n+1)\Delta
\le
\tilde\lambda_k(S^\star)+k(n+1)\Delta
\le
\lambda_k(S^\star)+2k(n+1)\Delta.
\]
Since $2k(n+1)\Delta\le \varepsilon M$, we get
$\lambda_k(\widetilde S)\le \OPT+\varepsilon M$.

It remains to bound the running time. The scaled rounded instance has integral
data and total processing time
\[
\sum_j \frac{\tilde p_j}{\Delta}
=O\left(\frac{kn^2}{\varepsilon}\right).
\]
By Theorem~\ref{thm:pseudopoly}, the exact algorithm on the scaled instance
therefore runs in time polynomial in $n$, $k$, and $1/\varepsilon$. The numbers
used by the rounded instance have polynomial bit length under exact rational
arithmetic. This proves the claimed running time.
\end{proof}

\subsubsection{An XP algorithm for fixed \texorpdfstring{$P$}{P}}
\label{sec:xp-p}

Assume that there are $P$ distinct processing times,
$\rho_1<\cdots<\rho_P$. We build a finite set of shifts that contains all
breakpoints needed in
Lemma~\ref{lem:breakpoint-reduction}.

Define
\[
\widehat{\mathcal L}_P
=
\left\{
\sum_{\ell=1}^P a_\ell\rho_\ell+p_j-d_j:
j\in\Jobs,\ a\in\{0,\ldots,n\}^P
\right\}.
\]

\begin{lemma}
\label{lem:fixed-P-candidate-shifts}
Every breakpoint from Lemma~\ref{lem:breakpoint-reduction} belongs to
$\widehat{\mathcal L}_P$. Moreover,
\[
|\widehat{\mathcal L}_P|=O(n^{P+1}).
\]
\end{lemma}

\begin{proof}
Consider a schedule $S$ and a job $j$. The start time of $j$ is the sum of the
processing times of the jobs scheduled before $j$. Since the possible
processing times are $\rho_1,\ldots,\rho_P$, this start time can be written as
$\sum_{\ell=1}^P a_\ell\rho_\ell$, where $a_\ell$ is the number of earlier jobs
with processing time $\rho_\ell$. Therefore
$C_j(S)-d_j=\sum_{\ell=1}^P a_\ell\rho_\ell+p_j-d_j$, which belongs to
$\widehat{\mathcal L}_P$.

There are $n$ choices for $j$ and at most $(n+1)^P$ choices for the vector $a$.
Hence $|\widehat{\mathcal L}_P|=O(n^{P+1})$.\qedhere
\end{proof}

\begin{algorithm}[htbp]
\caption{Fixed-$P$ algorithm for $k$-sum lateness}
\label{alg:fixed-P}
\DontPrintSemicolon
Generate $\widehat{\mathcal L}_P$\;
$v_{\mathrm{best}}\leftarrow+\infty$; $S_{\mathrm{best}}\leftarrow\textsc{null}$\;
\ForEach{$r\in\widehat{\mathcal L}_P$}{
  Solve the shifted total-tardiness instance with due dates $d_j+r$ using Theorem~\ref{thm:TT-fixed-P}\;
  Let $S_r$ be the returned schedule, and let $\Phi(r)$ be its value\;
  \If{$kr+\Phi(r)<v_{\mathrm{best}}$}{
    $v_{\mathrm{best}}\leftarrow kr+\Phi(r)$\;
    $S_{\mathrm{best}}\leftarrow S_r$\;
  }
}
\Return $S_{\mathrm{best}}$\;
\end{algorithm}

\begin{theorem}
\label{thm:fixed-P}
Algorithm~\ref{alg:fixed-P} returns an optimal schedule for $k$-sum lateness in
time
\[
O(Pn^{2P+5}).
\]
In particular, for fixed $P$, this is $O(n^{2P+5})$.
\end{theorem}

\begin{proof}
By Lemmas~\ref{lem:breakpoint-reduction} and
\ref{lem:fixed-P-candidate-shifts}, it is enough to enumerate the shifts in
$\widehat{\mathcal L}_P$. This set has size $O(n^{P+1})$.

For each shift, Theorem~\ref{thm:TT-fixed-P} computes $\Phi(r)$ in time
$O(Pn^{P+4})$. The total running time is therefore
\[
O(n^{P+1}\cdot Pn^{P+4})=O(Pn^{2P+5}).
\]
For fixed $P$, this is $O(n^{2P+5})$.

The algorithm returns the schedule with minimum value $kr+\Phi(r)$ among all
enumerated shifts, so it is optimal by~\eqref{eq:shift-tt}.
\end{proof}

\subsubsection{An XP algorithm for fixed \texorpdfstring{$D$}{D}}
\label{sec:xp-d-dual}

Assume that there are $D$ distinct due dates $\delta_1<\cdots<\delta_D$. Use the
notation from Section~\ref{sec:TT-fixed-D}. Thus
\[
J_i=\{j\in\Jobs:d_j=\delta_i\}
\]
is sorted as a list in the strict SPT order, $n_i=|J_i|$, and
\[
Q_i(h)=\sum_{\ell=1}^h p_{J_i[\ell]},
\qquad
Q(a)=\sum_{i=1}^D Q_i(a_i).
\]

Define
\[
\widehat{\mathcal L}_D
=
\{Q(a)-\delta_i:
a\in\prod_{\ell=1}^D\{0,\ldots,n_\ell\},\
i\in[D],\
a_i>0\}.
\]

\begin{lemma}
\label{lem:fixed-D-candidate-shifts}
The minimum in~\eqref{eq:shift-tt} is attained at some
$r\in\widehat{\mathcal L}_D$. Moreover,
\[
|\widehat{\mathcal L}_D|
\le
D\prod_{i=1}^D(n_i+1)
\le
D(n+1)^D.
\]
\end{lemma}

\begin{proof}
Let $r_0$ be an optimal shift in~\eqref{eq:shift-tt}. Choose an optimal
shifted total-tardiness schedule $S$ for this shift. By
Lemma~\ref{lem:spt-within-due-date}, we may choose $S$ so that it is a merge of
the ordered lists $J_1,\ldots,J_D$.

Now keep this schedule $S$ fixed and consider the function
\[
f_S(r)
=
kr+\sum_{j\in\Jobs}(C_j(S)-d_j-r)_+ .
\]
As in Lemma~\ref{lem:breakpoint-reduction}, there is a breakpoint $r'$ of
$f_S$ such that
\[
f_S(r')\le f_S(r_0).
\]
Since $S$ is optimal for the shift $r_0$,
\[
f_S(r_0)=kr_0+\Phi(r_0)=\OPT.
\]
Also,
\[
kr'+\Phi(r')\le f_S(r').
\]
Therefore $r'$ is also an optimal shift. It remains to show that $r'$ belongs
to $\widehat{\mathcal L}_D$.

Since $r'$ is a breakpoint of $f_S$, there is a job $j$ such that
$r'=C_j(S)-d_j$. Suppose $j\in J_i$. Because $S$ is a merge of the ordered
lists, the prefix of $S$ ending at $j$ contains the first $a_\ell$ jobs of each
list $J_\ell$, for some vector $a\in\prod_{\ell=1}^D\{0,\ldots,n_\ell\}$ with
$a_i>0$. Hence $C_j(S)=Q(a)$ and $d_j=\delta_i$. Thus
$r'=Q(a)-\delta_i\in\widehat{\mathcal L}_D$.

The size bound follows from the definition of $\widehat{\mathcal L}_D$: there
are $\prod_i(n_i+1)$ choices for $a$ and at most $D$ choices for $i$.
\end{proof}

\begin{algorithm}[htbp]
\caption{Fixed-$D$ algorithm for $k$-sum lateness}
\label{alg:fixed-D-dual}
\DontPrintSemicolon
Sort the distinct due dates as $\delta_1<\cdots<\delta_D$\;
\ForEach{$i\in[D]$}{
  Compute $J_i=\{j:d_j=\delta_i\}$ and sort $J_i$ by the strict \SPT{} order\;
}
Generate $\widehat{\mathcal L}_D$\;
$v_{\mathrm{best}}\leftarrow+\infty$; $S_{\mathrm{best}}\leftarrow\textsc{null}$\;
\ForEach{$r\in\widehat{\mathcal L}_D$}{
  Solve the shifted total-tardiness instance with due dates $d_j+r$ using Lemma~\ref{lem:TT-fixed-D}\;
  Let $S_r$ be the returned schedule, and let $\Phi(r)$ be its value\;
  \If{$kr+\Phi(r)<v_{\mathrm{best}}$}{
    $v_{\mathrm{best}}\leftarrow kr+\Phi(r)$\;
    $S_{\mathrm{best}}\leftarrow S_r$\;
  }
}
\Return $S_{\mathrm{best}}$\;
\end{algorithm}

\begin{theorem}
\label{thm:fixed-D-dual}
For fixed $D$, Algorithm~\ref{alg:fixed-D-dual} returns an optimal schedule for
$k$-sum lateness in time
\[
O\left(
D^2\left[\prod_{i=1}^D(n_i+1)\right]^2
\right)
\le
O(D^2(n+1)^{2D}),
\]
up to the time needed to sort the lists $J_i$.
\end{theorem}

\begin{proof}
By Lemma~\ref{lem:fixed-D-candidate-shifts}, it is enough to enumerate the
shifts in $\widehat{\mathcal L}_D$. This set has size at most
\[
D\prod_{i=1}^D(n_i+1).
\]
For each shift, Lemma~\ref{lem:TT-fixed-D} computes $\Phi(r)$ in time
\[
O\left(D\prod_{i=1}^D(n_i+1)\right).
\]
The total running time is therefore
\[
O\left(
D^2\left[\prod_{i=1}^D(n_i+1)\right]^2
\right)
\le
O(D^2(n+1)^{2D}).
\]
The algorithm returns the schedule with minimum value $kr+\Phi(r)$ among the
enumerated shifts, so it is optimal by~\eqref{eq:shift-tt}.
\end{proof}

    The remaining sections develop a different route based on the internal combinatorial
    structure of optimal schedules. This structural route is not needed for the shifted
    total-tardiness algorithms above, but it will be used to obtain improved XP bounds and
    an FPT algorithm for the parameter $D+k$.

\section{Combinatorial structure of optimal schedules}
\label{sec:structure}

This section presents a purely combinatorial route, which can be read
independently of the dual shifted-total-tardiness formulation in
Section~\ref{sec:dual}. We develop a canonical form for optimal schedules.
First, we show that some optimum respects a natural dominance order
(increasing schedules). We then fix a suitable top-$k$ set and prove that the
remaining jobs form consecutive due-date blocks ordered by due date. Finally,
we prove a suffix property inside each due-date class. These properties are
summarized in Theorem~\ref{thm:block-island} and are used later in the XP and
FPT algorithms.

None of the statements requires $D$ to be constant. The parameter $D$ becomes
relevant only when we turn these structures into explicit enumerations.

\subsection{Structured optimal solutions}

In $k$-sum tardiness, the penalty for a job completing before its due date is
truncated to zero, so such jobs can be scheduled earlier without changing the
objective value. In $k$-sum lateness, the penalty $C_j-d_j$ is not truncated.
Decreasing the completion time of any job strictly decreases its lateness.
Because every local reordering alters the lateness vector, establishing
structural properties requires explicit exchange arguments.

This section proves that there exists an optimal \emph{increasing} schedule (one
that respects a natural dominance order; see Definition~\ref{def:increasing})
that admits a block--island decomposition. Structurally, the schedule can be
written as
\[
S = I_0 \mid B_1 \mid I_1 \mid B_2 \mid I_2 \mid \cdots \mid B_D \mid I_D,
\]
where each block $B_i$ contains all jobs outside the top-$k$ set with the $i$-th
smallest due date, and each island $I_i$ is a (possibly empty) subsequence
consisting entirely of top-$k$ jobs. The islands $I_0, I_1, \dots, I_D$ represent
the top-$k$ subsequences that appear before the first block, between consecutive
blocks, and after the last block.

\subsubsection{Increasing schedules}

    We first isolate the perturbation subroutine used in the structural proofs.
    The perturbation is only an analysis device. The algorithms below use the original
    processing times, except that ties in \SPT{} order are broken by the fixed order from
    Section~\ref{sec:prelim}.
    If all processing times are distinct, no perturbation is needed.
    
    \begin{lemma}[Perturbation subroutine]
    \label{lem:symbolic-perturbation}
    Fix the tie-breaking order on jobs from Section~\ref{sec:prelim}. Let $\operatorname{rank}(j)\in\{1,\dots,n\}$
    be the rank of job $j$ in this order. For $\varepsilon>0$, let $\mathcal{I}^\varepsilon$ be the
    instance obtained from the original instance $\mathcal{I}$ by replacing
    $p_j$ by
    \[
    p_j^\varepsilon=p_j+\varepsilon \operatorname{rank}(j),
\]
and keeping all due dates unchanged. In particular, every perturbed processing
time is strictly positive. For sufficiently small $\varepsilon$, the following
properties hold:
    \begin{enumerate}[leftmargin=2.2em]
    \item the order of processing times in $\mathcal{I}^\varepsilon$ is exactly the lexicographic order
    by $(p_j,\operatorname{rank}(j))$;
    \item every schedule that is optimal for $\mathcal{I}^\varepsilon$ is optimal for $\mathcal{I}$;
    \item for every fixed schedule $S$, every top-$k$ set for $S$ in $\mathcal{I}^\varepsilon$ is also a
    top-$k$ set for $S$ in $\mathcal{I}$.
    \end{enumerate}
    If all processing times and due dates are integral, one may take
    $\varepsilon=1/(4kn(n+1))$.
    \end{lemma}
    
    \begin{proof}
    Set $R=\sum_{j=1}^n \operatorname{rank}(j)=n(n+1)/2$. For every schedule $S$ and job $j$, the perturbation
    increases the completion time of $j$ by a number in $[0,\varepsilon R]$. Hence
    \[
    L_j(S)\le L_j^\varepsilon(S)\le L_j(S)+\varepsilon R
    \]
    and therefore
    \[
    \lambda_k(S)
    \le
    \lambda_k^\varepsilon(S)
    \le
    \lambda_k(S)+k\varepsilon R .
    \]

    Let $\Delta_p$ be the minimum positive difference between two distinct original processing
    times, if such a difference exists, and set $\Delta_p=1$ otherwise. Let $\Delta_L$ be the
    minimum positive difference between two lateness values $L_i(S)$ and $L_j(S)$ over all
    schedules $S$ and jobs $i,j$, if such a difference exists, and set $\Delta_L=1$ otherwise.
    Let $\Delta_\lambda$ be the minimum positive difference between two objective values
    $\lambda_k(S)$ and $\lambda_k(T)$, if such a difference exists, and set
    $\Delta_\lambda=1$ otherwise. Choose $\varepsilon>0$ so that
    \[
    \varepsilon n < \Delta_p,\qquad
    2\varepsilon R < \Delta_L,\qquad
    k\varepsilon R < \Delta_\lambda .
    \]

    The first inequality implies that the order of the perturbed processing times is the
    lexicographic order by $(p_j,\operatorname{rank}(j))$. The second inequality implies that
    no strict lateness comparison for a fixed schedule is reversed by the perturbation: if
    $L_a(S)>L_b(S)$ in the original instance, then $L_a^\varepsilon(S)>L_b^\varepsilon(S)$ in
    the perturbed instance.

    Now let $X$ be a top-$k$ set for $S$ in $\mathcal{I}^\varepsilon$. If $X$ were not a top-$k$
    set for $S$ in the original instance, then there would exist $x\in X$ and $y\notin X$ such
    that $L_y(S)>L_x(S)$ in the original instance. This strict inequality would be preserved in
    the perturbed instance, giving $L_y^\varepsilon(S)>L_x^\varepsilon(S)$, contradicting that
    $X$ is a top-$k$ set for $S$ in $\mathcal{I}^\varepsilon$. Hence $X$ is also a top-$k$ set
    for $S$ in $\mathcal{I}$.

    It remains to prove preservation of optimality. If all schedules have the same original
    objective value, there is nothing to prove. Otherwise,
    let $S^\varepsilon$ be optimal for the perturbed instance, and let $S^*$ be optimal for the
    original instance. If $S^\varepsilon$ were not optimal for the original instance, then
    \[
    \lambda_k(S^\varepsilon)\ge \lambda_k(S^*)+\Delta_\lambda .
    \]
    Thus
    \[
    \lambda_k^\varepsilon(S^*)
    <
    \lambda_k(S^*)+\Delta_\lambda
    \le
    \lambda_k(S^\varepsilon)
    \le
    \lambda_k^\varepsilon(S^\varepsilon),
    \]
    contradicting the optimality of $S^\varepsilon$ for the perturbed instance.

    If all data are integral, distinct processing times and distinct lateness values differ by
    at least $1$, and distinct objective values differ by at least $1$. With
    $\varepsilon=1/(4kn(n+1))$, we have $\varepsilon n<1$, $2\varepsilon R\le 1/(2k)\le 1/2$,
    and $k\varepsilon R=1/8$. Hence the same three inequalities hold.
    \end{proof}
    
    Unless explicitly stated otherwise, the structural proofs below are carried out
after the perturbation fixed in Lemma~\ref{lem:symbolic-perturbation}, or without
perturbation if the processing times are already distinct and positive. To keep
the notation readable, we write $p_j$, $C_j(S)$, $L_j(S)$, and $\lambda_k(S)$ for
the quantities used in the structural proof. The order \SPT{} is the corresponding
strict processing-time order. Equivalently, it is the lexicographic order obtained
from the input processing times and the fixed tie-breaking order. Thus all strict
inequalities involving processing times in the structural proofs below are
interpreted in the perturbed instance. After applying
Lemma~\ref{lem:symbolic-perturbation}, the resulting schedule and top-$k$ set are
transferred back to the original instance, whose processing times may be zero.
The algorithms later use the lexicographic \SPT{} order, but all objective values
are evaluated with the input processing times.

    \begin{definition}[Dominance and increasing schedules]
    \label{def:increasing}
    Job $i$ is \emph{dominated} by job $j$ if $i\prec_p j$ and $d_i\le d_j$.
    A schedule $S$ is \emph{increasing} if whenever $i$ is dominated by $j$, job $i$ appears
    before $j$ in $S$.
    \end{definition}
    
    For a vector $a\in\mathbb{R}^n$, let $s_k(a)$ denote the sum of its $k$ largest coordinates.
    
    \begin{lemma}[Top-$k$ exchange]
    \label{lem:topk-exchange}
    Let $a,b\in\mathbb{R}^n$ coincide in all coordinates except possibly $i$ and $j$.
    If $b_i+b_j\le a_i+a_j$ and $\max\{b_i,b_j\}\le \max\{a_i,a_j\}$, then the sum of the
    $k$ largest coordinates of $b$ is at most the sum of the $k$ largest coordinates of $a$.
    \end{lemma}
    
    \begin{proof}
    Let $x\in\{0,1\}^n$ select the $k$ largest coordinates of $b$.
    If $x_i=x_j=0$, then $\sum_h x_h b_h=\sum_h x_h a_h\le s_k(a)$.
    If $x_i=x_j=1$, then
    $\sum_h x_h b_h\le \sum_h x_h a_h\le s_k(a)$ by the assumption on $b_i+b_j$.
    It remains to consider the case in which exactly one of $i,j$ is selected by $x$.
    Choose $y\in\{0,1\}^n$ equal to $x$ outside $\{i,j\}$ and selecting, among $i,j$, an index
    with larger $a$-coordinate.
    Then $\sum_h y_h=k$ and
    \[
    \sum_h x_h b_h
    \le \sum_{h\notin\{i,j\}} x_h a_h+\max\{a_i,a_j\}
    =\sum_h y_h a_h
    \le s_k(a).
    \qedhere\]
    \end{proof}
    
    \begin{lemma}[Dominated-pair exchange]
    \label{lem:dominated-exchange}
    Let $S$ be a schedule in which job $i$ appears before job $j$, with $p_i>p_j$ and
    $d_i\ge d_j$.
    Let $S'$ be obtained from $S$ by swapping the positions of $i$ and $j$.
    Then $\lambda_k(S')\le \lambda_k(S)$.
    \end{lemma}
    
    \begin{proof}
    Let $H$ be the set of jobs strictly between $i$ and $j$ in $S$, and let
    $p(H)=\sum_{h\in H}p_h$.
    The swap changes lateness values as follows:
    \[
    \begin{array}{ll}
    L_i(S')=L_i(S)+p(H)+p_j, & L_j(S')=L_j(S)-p(H)-p_i,\\
    L_h(S')=L_h(S)+p_j-p_i \qquad \forall\, h\in H, &
    L_h(S')=L_h(S) \qquad \forall\, h\notin H\cup\{i,j\}.
    \end{array}
    \]
    Let $a$ be the lateness vector of $S$ and let $c$ be the lateness vector of $S'$.
    Define $b$ by $b_i=c_i$, $b_j=c_j$, and $b_h=a_h$ for $h\notin\{i,j\}$.
    Since $c_h\le b_h$ for every $h$, the top-$k$ sum of $c$ is at most the top-$k$ sum of $b$.
    Moreover, $b_i+b_j<a_i+a_j$, because $p_i>p_j$. Since $d_i\ge d_j$, we have $b_i\le a_j$,
    and since $p_i>p_j$, we have $b_j<a_j$. Thus $\max\{b_i,b_j\}\le \max\{a_i,a_j\}$.
    By Lemma~\ref{lem:topk-exchange}, $s_k(b)\le s_k(a)$, so $\lambda_k(S')\le\lambda_k(S)$.
    \end{proof}
    
    \begin{theorem}[Existence of an increasing optimum]
    \label{thm:increasing-opt}
    There exists an optimal schedule that is increasing.
    \end{theorem}
    
    \begin{proof}
    In the perturbed instance, among all optimal schedules choose one, say $S$, that minimizes
    total lateness $\sum_j L_j(S)$.
    If $S$ is not increasing, then there are jobs $i,j$ such that $i$ is dominated by $j$ but
    $j$ appears before $i$ in $S$.
    Let $S'$ be obtained by swapping $j$ and $i$.
    Since $p_j>p_i$ and $d_j\ge d_i$ in the perturbed instance,
Lemma~\ref{lem:dominated-exchange} gives
$\lambda_k(S')\le\lambda_k(S)$, so $S'$ is also optimal.
    If $H$ is the set of jobs strictly between $j$ and $i$ in $S$, then
    \[
    \sum_h L_h(S')-\sum_h L_h(S)=(p_i-p_j)(|H|+1)<0,
    \]
    contradicting the choice of $S$.
    Therefore $S$ is increasing. By Lemma~\ref{lem:symbolic-perturbation}, $S$ is also optimal
    for the original instance.
    \end{proof}
    
\subsubsection{Proximity to \EDD{} outside the top-\texorpdfstring{$k$}{k} set}

The results in this subsection are stated and proved for the original instance.
The main ingredient is a proximity theorem: there exists an optimal schedule
close to the \EDD{} schedule in a strong combinatorial sense. We first recall
the notion of adjacent permutations and then prove the proximity theorem by
shifting due dates in an optimal schedule. The enumeration based on this
proximity result is given in Section~\ref{sec:xp-k}.
    
\paragraph{Adjacent permutations.}
    
    \begin{definition}[$m$-adjacent permutations]
    Given two schedules $S_1,S_2$ of $\Jobs$, we say that $S_2$ is \emph{$m$-adjacent}
    to $S_1$ if $S_2$ can be obtained from $S_1$ by relocating at most $m$ jobs.
    Equivalently, there exists a set $R\subseteq\Jobs$ with $|R|\le m$ such that after removing
    $R$ from both permutations, the remaining orders coincide.
    \end{definition}
    
    \begin{lemma}[Woeginger \cite{Woeginger1991}]
    \label{lem:adj-count}
    For any schedule $S$ and any $m\ge 0$, the number of $m$-adjacent permutations of $S$
    is $O(n^{2m})$.
    \end{lemma}
    
\paragraph{A due-date shifting transformation.}
    
    Fix an optimal schedule $S$ and let $S[k]=\{i_1,\dots,i_k\}$ denote the canonical top-$k$
    set of $S$ (with ties in lateness broken by $\prec_p$), ordered so that
    $L_{i_1}(S)\ge \cdots \ge L_{i_k}(S)$.
    We define a modified instance (same processing times) by shifting due dates:
    \begin{equation}
    \label{eq:shift}
    \tilde d_j =
    \begin{cases}
    d_j + L_j(S) & \text{if } j\in S[k],\\
    d_j + L_{i_k}(S) & \text{if } j\notin S[k].
    \end{cases}
    \end{equation}
    Note that, unlike tardiness, lateness values may be negative, so the shifts in~\eqref{eq:shift}
    may increase, decrease, or preserve due dates.
    
    Let $\tilde L_j(S)=C_j(S)-\tilde d_j$ and let $\tilde\lambda_k(S)$ be the
    top-$k$ sum of the modified instance.
    
    \begin{lemma}
    \label{lem:shift-properties}
    For the schedule $S$ used to define~\eqref{eq:shift}, we have:
    \begin{enumerate}[leftmargin=2.2em]
    \item $\tilde L_j(S)=0$ for all $j\in S[k]$,
    \item $\tilde L_j(S)\le 0$ for all $j\notin S[k]$,
    \item $\tilde\lambda_k(S)=0$.
    \end{enumerate}
    \end{lemma}
    
    \begin{proof}
    For $j\in S[k]$, by definition $\tilde d_j=d_j+L_j(S)$, hence
    $\tilde L_j(S)=C_j(S)-\tilde d_j=0$.
    For $j\notin S[k]$, we have $\tilde d_j=d_j+L_{i_k}(S)$, so
    $\tilde L_j(S)=L_j(S)-L_{i_k}(S)\le 0$ because $L_{i_k}(S)$ is the $k$-th
    largest lateness in $S$.
    The third claim follows immediately.
    \end{proof}
    
\paragraph{Proximity to \EDD{}.}
    
    \begin{theorem}[Proximity to \EDD{}]
    \label{thm:edd-proximity}
    For every optimal schedule $S$, there exists an optimal schedule $S^\star$ such that
    $S^\star$ is $(k-1)$-adjacent to the \EDD{} schedule for the original instance.
    \end{theorem}
    
    \begin{proof}
    Consider the shifted instance defined by~\eqref{eq:shift}.
    By Lemma~\ref{lem:shift-properties}, the maximum shifted lateness of $S$ is at most $0$.
    Let $E$ be the \EDD{} schedule for the original due dates $d$, and let $\widetilde E$ be
    the \EDD{} schedule for the shifted due dates $\tilde d$, using the \EDD{} tie-breaking
    rule fixed in Section~\ref{sec:prelim} in both cases.
    Since \EDD{} is optimal for maximum lateness \cite{Jackson1955}, we have
    $\max_j \tilde L_j(\widetilde E)\le 0$, and hence $\tilde L_j(\widetilde E)\le 0$ for every job $j$.
    
    We first compare the objective values in the original and shifted instances.
    Let $\widetilde E[k]$ denote the canonical top-$k$ set of $\widetilde E$ with respect to the
    original lateness values (breaking ties by $\prec_p$).
    Since completion times are unchanged by the due-date shift, for every job $j$ and every
    schedule $Q$, $L_j(Q)=\tilde L_j(Q)+(\tilde d_j-d_j)$.
    Therefore
    \begin{align*}
    \lambda_k(\widetilde E)
    &=\sum_{j\in \widetilde E[k]} L_j(\widetilde E)\\
    &=\sum_{j\in \widetilde E[k]} \tilde L_j(\widetilde E)
      + \sum_{j\in \widetilde E[k]}(\tilde d_j-d_j)\\
    &\le \sum_{j\in \widetilde E[k]}(\tilde d_j-d_j).
    \end{align*}
    By the definition of the shifted due dates,
    \begin{align*}
    \sum_{j\in \widetilde E[k]}(\tilde d_j-d_j)
    &=
    \sum_{j\in \widetilde E[k]\cap S[k]} L_j(S)
    +\sum_{j\in \widetilde E[k]\setminus S[k]} L_{i_k}(S)\\
    &=
    \sum_{j\in \widetilde E[k]\cap S[k]} L_j(S)
    + |S[k]\setminus \widetilde E[k]|\,L_{i_k}(S)\\
    &\le
    \sum_{j\in \widetilde E[k]\cap S[k]} L_j(S)
    + \sum_{j\in S[k]\setminus \widetilde E[k]} L_j(S)\\
    &=\lambda_k(S),
    \end{align*}
    where the second equality follows since
    $|S[k]\setminus \widetilde E[k]|=|\widetilde E[k]\setminus S[k]|$.
    The inequality holds because $i_k$ has minimum lateness among the jobs of $S[k]$.
    Thus $\lambda_k(\widetilde E)\le \lambda_k(S)$.
    
    It remains to compare $E$ and $\widetilde E$.
    Let $R=S[k]\setminus\{i_k\}$.
    Every job outside $R$ receives the same due-date shift, namely $L_{i_k}(S)$.
    Indeed, this is true for all jobs outside $S[k]$ by~\eqref{eq:shift}, and it is also true
    for $i_k$.
    Hence, for any two jobs $a,b\notin R$, we have
    $\tilde d_a\le \tilde d_b$ if and only if $d_a\le d_b$.
    With the fixed tie-breaking order, the restriction of $E$ to $\Jobs\setminus R$ is the same
    as the restriction of $\widetilde E$ to $\Jobs\setminus R$.
    Therefore, $\widetilde E$ is obtained from $E$ by relocating only the jobs in $R$.
    Since $|R|\le k-1$, $\widetilde E$ is $(k-1)$-adjacent to $E$. Setting
    $S^\star=\widetilde E$ proves the theorem.
    \end{proof}
\subsubsection{Block--island decomposition}
    
    Fix an optimal increasing schedule $S$ in the perturbed instance, and fix a
    top-$k$ set $X$ for $S$. For each due-date value $\delta$, define
    \[
    B_\delta(S,X)=\{j\in \Jobs\setminus X : d_j=\delta\}.
    \]
    We call the sets $B_\delta(S,X)$ the due-date blocks outside $X$. The maximal
    subsequences of jobs from $X$ between consecutive due-date blocks are called
    islands.
    Thus these blocks partition the jobs outside the fixed top-$k$ set $X$ according to
    their due date.
    
    \begin{lemma}[Strict block exchange]
    \label{lem:strict-block-exchange}
    Let $S$ be increasing, and let $X$ be a top-$k$ set for $S$.
    Suppose that $i,j\in \Jobs\setminus X$, that $i$ appears before $j$, that $i$ and $j$ are
    consecutive in the subsequence induced by $\Jobs\setminus X$, and that $d_i>d_j$.
    Let $S'$ be obtained from $S$ by removing $i$ and inserting it immediately after $j$.
    Then, in the perturbed instance, $\lambda_k(S')\le \lambda_k(S)$, the schedule $S'$ is
    increasing, and $\sum_h L_h(S')>\sum_h L_h(S)$.
    \end{lemma}
    
    \begin{proof}
    All processing times in this proof are the perturbed processing times. In
    particular, every job has strictly positive processing time, and the \SPT{}
    order is strict.
    Let $H$ be the set of jobs after $i$ and up to and including $j$ in $S$.
    Since $i$ and $j$ are consecutive in $\Jobs\setminus X$, every job in $H\setminus\{j\}$
    belongs to $X$.
    The reinsertion changes lateness values as follows:
    \[
    L_h(S')=L_h(S)-p_i \qquad \forall\, h\in H,\qquad
    L_q(S')=L_q(S) \qquad \forall\, q\notin H\cup\{i\},
    \]
    and
    \[
    L_i(S')=C_j(S)-d_i<C_j(S)-d_j=L_j(S).
    \]
    
\emph{Objective value.}
We first compare top-$k$ sums.
    If $h\ne i$ and $L_h(S)\ge L_j(S)$, then $L_h(S')\ge L_j(S')$.
    This is immediate if $h\in H$, since both sides decrease by $p_i$, and otherwise
    $L_h(S')=L_h(S)\ge L_j(S)>L_j(S')$.
    In particular, the $k$ jobs of the fixed top-$k$ set $X$ have lateness at
    least $L_j(S')$ in $S'$. Therefore there exists a top-$k$ set for $S'$ that
    does not contain $j$. Indeed, if a top-$k$ set $K_0$ selected $j$, then
    $K_0$ cannot contain all jobs of $X$, because $|X|=k$ and $j\notin X$.
    Choose $x\in X\setminus K_0$. Since $L_x(S')\ge L_j(S')$, replacing $j$ by
    $x$ gives another top-$k$ set for $S'$. Fix such a set $K$ with
    $j\notin K$.
    If $i\notin K$, then
    \[
    \lambda_k(S')=\sum_{h\in K}L_h(S')\le \sum_{h\in K}L_h(S)\le \lambda_k(S).
    \]
    If $i\in K$, then
    \[
    \lambda_k(S')
    \le \sum_{h\in K\setminus\{i\}}L_h(S)+L_j(S)
    \le \lambda_k(S),
    \]
    because $(K\setminus\{i\})\cup\{j\}$ has size $k$.
    
\emph{Increasing property.}
It remains to verify that the reinsertion preserves the increasing property.
    For every $h\in H$, we must have $d_h<d_i$.
    This is immediate for $h=j$.
    If $h\in H\setminus\{j\}$ and $d_h\ge d_i$, then
    $L_h(S)=C_h(S)-d_h<C_j(S)-d_h\le C_j(S)-d_j=L_j(S)$,
    contradicting $h\in X$ and $j\notin X$.
    Since $S$ is increasing and $i$ appears before every job in $H$, we also have
    $p_h>p_i$ for every $h\in H$.
    Indeed, if $p_h<p_i$, then $h$ is dominated by $i$, and hence $h$ would have to appear
    before $i$ in $S$.
    Thus the reinsertion moves $i$ across jobs that are incomparable with it, and no dominance
    relation is violated.
    All other relative orders are unchanged, so $S'$ is increasing.
    
    Finally,
    \[
    \sum_h L_h(S')-\sum_h L_h(S)
    =\sum_{h\in H}p_h-|H|p_i>0,
    \]
    because $p_h>p_i$ for every $h\in H$.
    \end{proof}
    
    \begin{lemma}[Equal-deadline block exchange]
    \label{lem:equal-deadline-block-exchange}
    Let $S$ be increasing, and let $X$ be a top-$k$ set for $S$.
    Suppose that $i,j\in \Jobs\setminus X$, that $i$ appears before $j$, that $i$ and $j$ are
    consecutive in the subsequence induced by $\Jobs\setminus X$, that $d_i=d_j$, and that
    $i$ and $j$ are not consecutive in $S$.
    Let $S'$ be obtained from $S$ by removing $i$ and inserting it immediately before $j$.
    Then, in the perturbed instance, $\lambda_k(S')\le \lambda_k(S)$, the schedule $S'$ is
    increasing, and $\sum_h L_h(S')>\sum_h L_h(S)$.
    \end{lemma}
    
    \begin{proof}
    All processing times in this proof are the perturbed processing times. In
    particular, every job has strictly positive processing time, and the \SPT{}
    order is strict.
    Let $B$ be the set of jobs strictly between $i$ and $j$ in $S$.
    Since $i$ and $j$ are consecutive in $\Jobs\setminus X$, we have $B\subseteq X$.
    Also $B$ is non-empty, because $i$ and $j$ are not consecutive in $S$.
    The reinsertion changes lateness values as follows:
    \[
    L_h(S')=L_h(S)-p_i \qquad \forall\, h\in B,\qquad
    L_q(S')=L_q(S) \qquad \forall\, q\notin B\cup\{i\},
    \]
    and
    \[
    L_i(S')=C_j(S)-p_j-d_i=L_j(S)-p_j.
    \]
    The completion time of $j$ is unchanged by this reinsertion.
    
    Since $S$ is increasing and $d_i=d_j$, we have $p_i<p_j$.
    Indeed, if $p_j<p_i$, then $j$ would be dominated by $i$, contradicting the order
    $i<_S j$ in an increasing schedule.
    For every $h\in B$, the inequality $h\in X$ and $j\notin X$ gives
    $L_h(S)\ge L_j(S)$.
    Hence
    \[
    L_h(S')=L_h(S)-p_i\ge L_j(S)-p_i>L_j(S)-p_j=L_i(S')
    \qquad (h\in B).
    \]
    The same strict inequality $L_h(S')>L_i(S')$ holds for $h\in X\setminus B$, since then
    $L_h(S')=L_h(S)\ge L_j(S)>L_i(S')$.
    Thus the $k$ jobs of the fixed top-$k$ set $X$ all have lateness strictly
    larger than $L_i(S')$ in $S'$. Therefore there exists a top-$k$ set $K$ for
    $S'$ with $i\notin K$. Indeed, if a top-$k$ set $K_0$ selected $i$, then
    $K_0$ cannot contain all jobs of $X$, because $|X|=k$ and $i\notin X$.
    Choose $x\in X\setminus K_0$. Since $L_x(S')>L_i(S')$, replacing $i$ by
    $x$ gives another top-$k$ set for $S'$. For such a set $K$,
    \[
    \lambda_k(S')=\sum_{h\in K}L_h(S')\le \sum_{h\in K}L_h(S)\le \lambda_k(S).
    \]
    
    It remains to verify that $S'$ is increasing.
    For every $h\in B$, we must have $d_h<d_i$.
    Indeed, if $d_h\ge d_i=d_j$, then
    $L_h(S)=C_h(S)-d_h<C_j(S)-d_h\le C_j(S)-d_j=L_j(S)$,
    contradicting $h\in X$ and $j\notin X$.
    Since $S$ is increasing and $i$ appears before every job in $B$, we also have
    $p_h>p_i$ for every $h\in B$.
    Indeed, if $p_h<p_i$, then $h$ would be dominated by $i$, and hence $h$ would have to
    appear before $i$ in $S$.
    Thus $i$ is incomparable with every job in $B$.
    The reinsertion moves $i$ only across the jobs of $B$, and $i$ remains before $j$.
    All other relative orders are unchanged, so $S'$ is increasing.
    
    Finally,
    \[
    \sum_h L_h(S')-\sum_h L_h(S)
    =\sum_{h\in B}p_h-|B|p_i>0,
    \]
    because $B$ is non-empty and $p_h>p_i$ for every $h\in B$.
    \end{proof}
    
    A key structural result shows that we may assume the blocks outside $X$ occupy consecutive
    positions in the schedule, and that these blocks follow due-date order.

    Recall that for a due date value $\delta$ we write
    $B_\delta(S,X)=\{j\in \Jobs\setminus X : d_j=\delta\}$.
    
    \begin{theorem}[Consecutive increasing blocks]
    \label{thm:consecutive-blocks}
    There exists an optimal increasing schedule $S$ and a top-$k$ set $X$ for $S$ such that:
    \begin{enumerate}[leftmargin=2.2em]
    \item For each due-date value $\delta$, the jobs in $B_\delta(S,X)$ occupy consecutive positions in $S$.
    \item If two blocks $B_{\delta}(S,X)$ and $B_{\delta'}(S,X)$ are non-empty and $B_{\delta}(S,X)$ is to the left
    of $B_{\delta'}(S,X)$, then $\delta<\delta'$.
    \end{enumerate}
    \end{theorem}
    
    \begin{proof}
    In the perturbed instance, there exists an optimal increasing schedule by the exchange
    argument in Theorem~\ref{thm:increasing-opt}. Among all pairs consisting of a perturbed
    optimal increasing schedule and a top-$k$ set for it, choose one, say $(S,X)$, that
    maximizes the perturbed total lateness $\sum_h L_h(S)$.
    
    We first show that the subsequence induced by $\Jobs\setminus X$ is ordered by nondecreasing
    due dates.
    If not, there are jobs $i,j\in \Jobs\setminus X$ that are consecutive in this subsequence,
    with $i$ before $j$ and $d_i>d_j$.

    Let $S'$ be obtained by removing $i$ and inserting it immediately after $j$.
    By Lemma~\ref{lem:strict-block-exchange}, $S'$ is increasing and satisfies
    $\lambda_k(S')\le \lambda_k(S)$, so $S'$ is also optimal. Choose any top-$k$
    set $X'$ for $S'$ in the perturbed instance. Then $(S',X')$ is one of the
    pairs over which the maximum was taken. The same lemma gives
    $\sum_h L_h(S')>\sum_h L_h(S)$, contradicting the choice of $(S,X)$.

    It remains to show that each block $B_{\delta}(S,X)$ is consecutive.
    Since the jobs in $\Jobs\setminus X$ are ordered by due date, if the jobs of some due-date
    class are not consecutive in $S$, then there are jobs $i,j\in \Jobs\setminus X$ that are
    consecutive in the subsequence induced by $\Jobs\setminus X$, have $d_i=d_j$, and are not
    consecutive in $S$.

    Lemma~\ref{lem:equal-deadline-block-exchange} gives an increasing schedule
    $S'$ with $\lambda_k(S')\le\lambda_k(S)$ and
    $\sum_h L_h(S')>\sum_h L_h(S)$. Hence $S'$ is optimal. Choosing any top-$k$
    set $X'$ for $S'$ gives a feasible pair $(S',X')$ with larger total lateness
    than $(S,X)$, another contradiction.

    Therefore all blocks $B_{\delta}(S,X)$ are consecutive and appear in increasing due-date
    order in the perturbed instance. By Lemma~\ref{lem:symbolic-perturbation}, the schedule is
    optimal for the original instance, and $X$ is also a top-$k$ set for the original lateness
    vector.
    \end{proof}
    
    \begin{theorem}[Block-island decomposition]
\label{thm:block-island}
There exists an optimal schedule $S$ for the original instance and a top-$k$
set $X$ for $S$ with the following properties.

\begin{enumerate}
\item $S$ is increasing with respect to the lexicographic \SPT{} order fixed in
Section~\ref{sec:prelim}.
\item Let $\delta_1<\cdots<\delta_D$ be the distinct due dates. Then $S$ can be
written as
\[
S = I_0 \mid B_1 \mid I_1 \mid B_2 \mid I_2 \mid \cdots \mid B_D \mid I_D,
\]
where, for each $i\in[D]$, the block $B_i$ consists exactly of the jobs
\[
B_i=\{j\in \Jobs\setminus X:d_j=\delta_i\}.
\]
Each island $I_i$ consists only of jobs from $X$ and may contain jobs with
different due dates. A block $B_i$ or island $I_i$ may be empty.
\item For each due-date class $J_i=\{j\in \Jobs:d_j=\delta_i\}$, sorted by the
lexicographic \SPT{} order, the set $J_i\setminus X$ is a prefix of $J_i$, and
$X\cap J_i$ is the corresponding suffix.
\end{enumerate}
\end{theorem}
\begin{proof}
We use the perturbation only inside this proof. Apply
Theorem~\ref{thm:consecutive-blocks} in the perturbed instance. We obtain an
optimal increasing schedule $S$ and a top-$k$ set $X$ such that the jobs outside
$X$ form consecutive blocks grouped by due date, and these blocks appear in
increasing due-date order.

Let $\delta_1<\cdots<\delta_D$ be the distinct due dates, and set
\[
B_i=\{j\in \Jobs\setminus X:d_j=\delta_i\}
\]
for each $i\in[D]$. The maximal subsequences of jobs from $X$ before the first
block, between two consecutive blocks, and after the last block are the islands
$I_0,I_1,\ldots,I_D$. This gives the stated block--island decomposition.

It remains to prove the suffix property. Fix a due-date class $J_i$. Let
$x\in X\cap J_i$ and $y\in J_i\setminus X$. Since $X$ is a top-$k$ set for $S$,
we have $L_x(S)\ge L_y(S)$. Since $d_x=d_y$, this implies
$C_x(S)\ge C_y(S)$. In the perturbed instance all processing times are strictly
positive. Hence $x$ cannot appear before $y$, because then $C_x(S)<C_y(S)$.
Therefore every job of $J_i\setminus X$ appears before every job of $X\cap J_i$.
Since $S$ is increasing, the jobs in one due-date class appear in the
lexicographic \SPT{} order. Thus $J_i\setminus X$ is a prefix of $J_i$, and
$X\cap J_i$ is the corresponding suffix.

By Lemma~\ref{lem:symbolic-perturbation}, the same schedule $S$ is optimal for
the original instance, and the same set $X$ is a top-$k$ set for $S$ in the
original instance. The remaining properties are order properties of this same
pair $(S,X)$. The block positions are unchanged because the schedule and the set
$X$ are unchanged, and due dates are not perturbed. The suffix property is stated
with respect to the lexicographic \SPT{} order, which is exactly the strict
processing-time order of the perturbed instance. Therefore, the block--island
decomposition and the prefix/suffix property obtained in the perturbed instance
are precisely the three properties stated for the original instance.
\end{proof}

\subsection{Algorithmic consequences}
    \label{sec:xp-structural}
    
    We now give algorithmic consequences of the structural results proved above.
    The algorithms in this section are stated for the original processing times and due dates.
    The perturbation subroutine is used only in the analysis, to argue that there exists an
    optimal schedule that satisfies the required structural properties.

\subsubsection{XP algorithms for fixed \texorpdfstring{$k$}{k}}

\paragraph{A proximity-based algorithm (cf.~Woeginger).}
Theorem~\ref{thm:edd-proximity} implies that there exists an optimal schedule that is
$(k-1)$-adjacent to the \EDD{} schedule (in the sense of relocating at most $k-1$ jobs).
Thus, for fixed $k$, we may enumerate all $(k-1)$-adjacent schedules and keep the best
one. By Lemma~\ref{lem:adj-count}, there are $O(n^{2(k-1)})$ such schedules, and each
candidate schedule can be evaluated in $O(n)$ time by computing its lateness vector and
selecting the $k$ largest values. This yields an $O(n^{2k-1})$-time XP algorithm.
This is analogous in spirit to Woeginger's proximity approach for $k$-sum tardiness.

\paragraph{An improved XP algorithm for fixed \texorpdfstring{$k$}{k}.}
\label{sec:xp-k}

We give an XP algorithm for fixed $k$. The algorithm guesses the top-$k$ set
$X$ and its order in the schedule. For this guess, the set
$Y=\Jobs\setminus X$ is ordered by increasing due date, breaking ties by
$\prec_p$. The algorithm also guesses a job $x_q\in X$ with minimum
lateness inside $X$, and the number of jobs of $Y$ that appear before it.
These guesses fix a threshold.  A dynamic program then computes the best
interleaving of the two orders.

For an ordered tuple $X=(x_1,\ldots,x_k)$, let
$Y=(y_1,\ldots,y_m)$, where $m=n-k$, be the order of $\Jobs\setminus X$
defined above.  Let $P_X(0)=P_Y(0)=0$,
\[
P_X(i)=\sum_{a=1}^i p_{x_a},
\qquad
P_Y(j)=\sum_{b=1}^j p_{y_b}.
\]

Fix $q\in[k]$ and $b\in\{0,\ldots,m\}$.  The intended meaning is
that $x_q$ has minimum lateness among the jobs of $X$, and exactly $b$ jobs
of $Y$ appear before $x_q$.  Define
\[
\theta=P_X(q)+P_Y(b)-d_{x_q}.
\]
Thus $x_q$ has lateness $\theta$ whenever the two orders are preserved and
$x_q$ is placed after exactly $b$ jobs of $Y$.

For this fixed choice of $(X,q,b)$, define a table $M[i,j]$, with
$0\le i\le k$ and $0\le j\le m$.  The entry $M[i,j]$ is the minimum value of
$\sum_{a=1}^i L_{x_a}$ over all partial schedules that place exactly
$x_1,\ldots,x_i$ and $y_1,\ldots,y_j$, preserve the two orders, and satisfy
the following conditions:
each placed job of $Y$ has lateness at most $\theta$;
each placed job $x_a$ with $a\ne q$ has lateness at least $\theta$;
$x_q$, if placed, is placed after exactly $b$ jobs of $Y$; and, if $x_q$
has not yet been placed, then at most $b$ jobs of $Y$ have been placed.
If no such partial schedule exists, then $M[i,j]=+\infty$.

The table is initialized by $M[0,0]=0$ and $M[i,j]=+\infty$ for all other
entries.  It is filled by forward transitions.  From a state $(i,j)$ with
$M[i,j]<+\infty$, the current time is $P_X(i)+P_Y(j)$.

If $j<m$, we may append $y_{j+1}$, unless this would place more than $b$ jobs
of $Y$ before $x_q$.  Thus the transition is allowed only when
$i\ge q$ or $j<b$.  In that case, let
\[
\ell_Y=P_X(i)+P_Y(j+1)-d_{y_{j+1}}.
\]
If $\ell_Y\le\theta$, set
\[
M[i,j+1]\leftarrow \min\{M[i,j+1],M[i,j]\}.
\]

If $i<k$, we may append $x_{i+1}$.  Let
\[
\ell_X=P_X(i+1)+P_Y(j)-d_{x_{i+1}}.
\]
If $i+1=q$, the transition is allowed only when $j=b$.  If $i+1\ne q$, the
transition is allowed only when $\ell_X\ge\theta$.  Whenever the transition is
allowed, set
\[
M[i+1,j]\leftarrow \min\{M[i+1,j],M[i,j]+\ell_X\}.
\]

Algorithm~\ref{alg:xp-k-dp} implements these transitions.

\begin{algorithm}[htbp]
\caption{XP algorithm for $k$-sum lateness}
\label{alg:xp-k-dp}
\DontPrintSemicolon
$v_{\mathrm{best}}\leftarrow +\infty$\;
\ForEach{ordered $k$-tuple $X=(x_1,\ldots,x_k)$ of distinct jobs}{
  $Y\leftarrow \Jobs\setminus\{x_1,\ldots,x_k\}$, sorted by increasing due date, breaking ties by $\prec_p$\;
  Write $Y=(y_1,\ldots,y_m)$, where $m=n-k$\;
  Compute $P_X(i)$ for $0\le i\le k$ and $P_Y(j)$ for $0\le j\le m$\;
  \ForEach{$(q,b)$ with $q\in[k]$ and $b\in\{0,\ldots,m\}$}{
    $\theta\leftarrow P_X(q)+P_Y(b)-d_{x_q}$\;
    Set all entries of $M[0..k,0..m]$ to $+\infty$ and then update $M[0,0]\leftarrow 0$\;
    \ForEach{$(i,j)$ with $0\le i\le k$, $0\le j\le m$, in nondecreasing order of $i+j$}{
      \If{$M[i,j]<+\infty$}{
        \If{$j<m$ and $(i\ge q$ or $j<b)$}{
          $\ell_Y\leftarrow P_X(i)+P_Y(j+1)-d_{y_{j+1}}$\;
          \lIf{$\ell_Y\le\theta$}{$M[i,j+1]\leftarrow \min\{M[i,j+1],M[i,j]\}$}
        }
        \If{$i<k$}{
          $\ell_X\leftarrow P_X(i+1)+P_Y(j)-d_{x_{i+1}}$\;
          \lIf{$(i+1=q$ and $j=b)$ or $(i+1\ne q$ and $\ell_X\ge\theta)$}{$M[i+1,j]\leftarrow \min\{M[i+1,j],M[i,j]+\ell_X\}$}
        }
      }
    }
    $v_{\mathrm{best}}\leftarrow \min\{v_{\mathrm{best}},M[k,m]\}$\;
  }
}
\Return $v_{\mathrm{best}}$\;
\end{algorithm}

\begin{theorem}
\label{thm:xp-k-improved}
For fixed $k\ge 1$, the optimal value of $k$-sum lateness can be computed in
$O(k^2n^{k+2})$ time. An optimal schedule can be recovered within the same
asymptotic time by storing predecessors.
\end{theorem}

\begin{proof}
Fix one choice of $(X,q,b)$.  For each state $(i,j)$, let $\mathcal{F}_{i,j}$ be the set of partial schedules used
in the definition of $M[i,j]$. Thus each schedule in $\mathcal{F}_{i,j}$ schedules exactly
$x_1,\ldots,x_i$ and $y_1,\ldots,y_j$, preserves the two orders, and satisfies the
threshold and position conditions stated above. In particular, if $i<q$, then
$j\le b$.  For $S\in\mathcal F_{i,j}$, let
$c(S)=\sum_{a=1}^i L_{x_a}(S)$.

We prove that the table satisfies
\[
M[i,j]=\min\{c(S):S\in\mathcal F_{i,j}\},
\]
with value $+\infty$ if $\mathcal F_{i,j}$ is empty.  The proof is by
induction on $i+j$.

The claim is clear for $(0,0)$.  Let $(i,j)$ be a state with $i+j>0$, and
assume the claim for smaller values of $i+j$.  Every schedule in
$\mathcal F_{i,j}$ has completion time $P_X(i)+P_Y(j)$.  Hence the lateness
of the last job is determined by the state and by the identity of that last
job.

Let $S\in\mathcal F_{i,j}$.  If the last job of $S$ is $y_j$, then deleting
$y_j$ gives a schedule in $\mathcal F_{i,j-1}$.  The deleted job has lateness
\[
P_X(i)+P_Y(j)-d_{y_j}.
\]
Thus the condition for appending $y_j$ is exactly the condition checked by
the first transition.  Since $y_j\notin X$, this transition adds no cost.
Therefore the best schedule in this case is obtained by the first transition,
when that transition is allowed.

If the last job of $S$ is $x_i$, then deleting $x_i$ gives a schedule in
$\mathcal F_{i-1,j}$.  The deleted job has lateness
\[
\ell_X=P_X(i)+P_Y(j)-d_{x_i}.
\]
If $i=q$, then $x_i=x_q$ must be placed after exactly $b$ jobs of $Y$, so
$j=b$; in this case $\ell_X=\theta$. If $i\ne q$, the required condition is
$\ell_X\ge\theta$. These are exactly the conditions checked by the second
transition. This transition adds exactly $\ell_X$ to the cost. Therefore
the best schedule in this case is obtained by the second transition, when
that transition is allowed.

Every nonempty schedule in $\mathcal F_{i,j}$ ends either with $y_j$ or with
$x_i$.  The two transitions cover these two cases.  Conversely, every allowed
transition appends the next job in one of the two fixed orders, checks the
required condition for that job, and preserves the condition on the position
of $x_q$.  Hence every allowed transition produces a schedule in the
corresponding set $\mathcal F_{i,j}$ with the value assigned by the table.
This proves the claim.

Taking $(i,j)=(k,m)$, the value $M[k,m]$ is the minimum of
$\sum_{x\in X}L_x$ over all schedules consistent with the fixed choice
$(X,q,b)$.  Every such schedule satisfies $L_x\ge\theta$ for all $x\in X$ and
$L_y\le\theta$ for all $y\in Y$.  Since $|X|=k$, the set $X$ is a top-$k$ set,
allowing ties at value $\theta$.  Therefore $M[k,m]$ is the $k$-sum lateness
value of the best schedule consistent with $(X,q,b)$.

It remains to show that one choice considered by the algorithm contains an
optimum. By Theorem~\ref{thm:block-island}, there is an optimal
schedule $S^*$ and a top-$k$ set $X^*$ such that $S^*$ is increasing and
$(S^*,X^*)$ admits a block--island decomposition. List the jobs of $X^*$ in
their order in $S^*$ as $X=(x_1,\ldots,x_k)$, and let $Y^*=\Jobs\setminus X^*$.

The block--island decomposition orders the jobs of $Y^*$ by nondecreasing due
date.  Since $S^*$ is increasing, jobs of $Y^*$ with the same due date appear
in \SPT{} order, with ties broken by $\prec$.  Hence the order of $Y^*$ in
$S^*$ is the order used by the algorithm for this choice of $X$.

Let $x_q$ be a job of minimum lateness inside $X^*$, and let $b$ be the number
of jobs of $Y^*$ placed before $x_q$ in $S^*$.  The algorithm enumerates this
choice of $(X,q,b)$.  For this choice, $\theta=L_{x_q}(S^*)$.  Since $X^*$ is a
top-$k$ set, all jobs of $X^*$ have lateness at least $\theta$, and all jobs of
$Y^*$ have lateness at most $\theta$.  Thus $S^*$ is feasible for the dynamic
program for this choice.  Hence the value computed for this choice is at most
$\lambda_k(S^*)$.

Conversely, every finite value computed by the dynamic program is the
$k$-sum lateness value of an actual schedule.  Hence no computed value is
smaller than the optimum.  Therefore Algorithm~\ref{alg:xp-k-dp} returns the
optimal value.

There are $n(n-1)\cdots(n-k+1)=O(n^k)$ ordered choices for $X$.  For each
choice there are $k$ choices for $q$ and $m+1=O(n)$ choices for $b$.  Each
dynamic program has $O(kn)$ states and constant-time transitions.  Thus the
total running time is
$O(n^k\cdot k\cdot n\cdot kn)=O(k^2 n^{k+2})$.  The sorting and the
construction of each sequence $Y$ are absorbed in this bound.  Storing
predecessors gives an optimal schedule within the same asymptotic time.
\end{proof}

\subsubsection{Consequence for top-\texorpdfstring{$k$}{k} tardiness}

The same algorithm gives the corresponding bound for $k$-sum tardiness. This
answers Woeginger's question of whether the problem admits an $O(n^{k+c})$-time
algorithm for a constant~$c$.

\begin{corollary}
\label{cor:woeginger-tardiness}
For fixed $k$, $k$-sum tardiness on one machine can be solved in
time
\[
O(k^2 n^{k+2}),
\]
where the asymptotic notation is with respect to $n$.
\end{corollary}

\begin{proof}
Given an instance of $k$-sum tardiness on $n$ jobs, add $k$ dummy
jobs with processing time $0$ and due date $0$. By
Proposition~\ref{prop:zero-dummy-reduction}, solving $k$-sum lateness on
the augmented instance is equivalent to solving the original $k$-sum tardiness
instance.

The augmented instance has $n+k$ jobs. Applying
Theorem~\ref{thm:xp-k-improved} gives running time
\[
O(k^2(n+k)^{k+2}).
\]
For fixed $k$, this is
\[
O(k^2n^{k+2})
\]
as a function of $n$. Restricting the returned augmented schedule to the
original jobs gives an optimal schedule for the original tardiness instance.
\end{proof}

\subsubsection{FPT algorithm for \texorpdfstring{$D+k$}{D+k}}
\label{sec:fpt-kd-dp}

This subsection gives an FPT algorithm for the parameter $D+k$, where $D$ is the number of distinct due dates. It is a direct refinement of the XP algorithm for fixed $k$ presented in Section~\ref{sec:xp-k}, taking advantage of the block--island decomposition to reduce the search space.

Let $\delta_1<\cdots<\delta_D$ be the distinct due dates. For each $i\in[D]$, let $J_i=\{j\in\Jobs:d_j=\delta_i\}$, let $n_i=|J_i|$, and write $J_i=(j_{i,1},\ldots,j_{i,n_i})$ for this class sorted by \SPT{} (with ties broken by $\prec$). We compute the prefix sums $A_i(t)=\sum_{s=1}^t p_{j_{i,s}}$ for $0\le t\le n_i$.

By Theorem~\ref{thm:block-island}, there is an optimal schedule with a top-$k$ set $X$ such that, for each $i\in[D]$, the set $X\cap J_i$ is a suffix of $J_i$ in lexicographic \SPT{} order, and $B_i=J_i\setminus X$ is the corresponding prefix block.

Instead of guessing the complete order of all jobs as we did in the XP algorithm, we now only need to guess the sizes $c_i = |X \cap J_i|$ and the relative order of the jobs in $X$. We do this by enumerating words $\sigma \in [D]^k$, which encode (in order) the due-date classes of the $k$ selected jobs. Let $c_i(\sigma)$ be the number of occurrences of symbol $i$ in $\sigma$. If $c_i(\sigma) > n_i$ for some $i$, the guess is discarded. Otherwise, define $X_i(\sigma)$ as the suffix of $J_i$ of size $c_i(\sigma)$, define the guessed top-$k$ set as $X(\sigma)=\bigcup_{i=1}^D X_i(\sigma)$, and define $B_i(\sigma)=J_i\setminus X_i(\sigma)$.

Given such a guess $\sigma$, we also need a concrete order of the jobs in $X(\sigma)$. We denote this ordered list by $\vec X(\sigma)=(x_1,\ldots,x_k)$, where the underlying set is $X(\sigma)=\{x_1,\ldots,x_k\}$. The word $\sigma$ determines $x_1,\ldots,x_k$ as follows. For each class $i$, list the jobs of $X_i(\sigma)$ in lexicographic \SPT{} order. Then assign these jobs, in that order, to the positions $a$ with $\sigma_a=i$, also taken in increasing order. This defines one job $x_a$ for every position $a\in[k]$. Let $P_X(a)=\sum_{b=1}^a p_{x_b}$ and $P_B(h)=\sum_{i=1}^h p(B_i(\sigma))$, where $p(B_i(\sigma))=\sum_{j\in B_i(\sigma)}p_j=A_i(n_i-c_i(\sigma))$.

Just as in the XP algorithm, we guess a position $q\in[k]$ for the job $x_q$ that achieves the minimum lateness within $X(\sigma)$. Instead of guessing the number of individual jobs outside the top-$k$ set that appear before $x_q$, we guess an index $h_0\in\{0,\ldots,D\}$: the DP will place exactly the block transitions $B_1,\ldots,B_{h_0}$ before placing $x_q$. Empty blocks are treated as zero-length block transitions, so $h_0$ is an index in the ordered block sequence, not the number of non-empty blocks before $x_q$ in the expanded schedule. This fixes the minimum lateness threshold $\theta = P_X(q) + P_B(h_0) - d_{x_q}$ of the top-$k$ set.

For this fixed choice of $(\sigma,q,h_0)$, define a table $M[a,h]$, with $0\le a\le k$ and $0\le h\le D$. The entry $M[a,h]$ is the minimum total lateness contributed by the placed jobs of $X(\sigma)$ over all partial schedules that place $x_1,\ldots,x_a$ and have processed the block transitions $B_1,\ldots,B_h$, preserve the relative orders of the $x$-jobs and of the block sequence, and satisfy the following conditions: each processed non-empty block has maximum lateness at most $\theta$; each placed job $x_m$ with $m\ne q$ has lateness at least $\theta$; $x_q$, if placed, is placed after exactly the block transitions $B_1,\ldots,B_{h_0}$; and, if $x_q$ has not yet been placed, then only block transitions with index at most $h_0$ have been processed. If no such partial schedule exists, $M[a,h]=+\infty$.

The table is initialized by $M[0,0]=0$ and $M[a,h]=+\infty$ otherwise. From a state $(a,h)$ with $M[a,h]<+\infty$, the current time is $P_X(a)+P_B(h)$. 

If $h<D$, we may append $B_{h+1}$, unless this would place more than $h_0$ blocks before $x_q$. Thus the transition is allowed only when $a\ge q$ or $h<h_0$. If $B_{h+1}$ is non-empty, we also require its maximum lateness $P_X(a)+P_B(h+1)-\delta_{h+1}$ to be at most $\theta$. If allowed, we update $M[a,h+1]\leftarrow \min\{M[a,h+1],M[a,h]\}$.

If $a<k$, we may append $x_{a+1}$. Let $\ell_X = P_X(a+1)+P_B(h)-d_{x_{a+1}}$. If $a+1=q$, the transition is allowed only when $h=h_0$. If $a+1\ne q$, it is allowed only when $\ell_X\ge\theta$. Whenever allowed, we update $M[a+1,h]\leftarrow \min\{M[a+1,h],M[a,h]+\ell_X\}$.

The minimum of $M[k,D]$ over all valid guesses gives the optimal $k$-sum lateness. The full procedure is shown in Algorithm~\ref{alg:fpt-Dk-dp}. The dynamic program evaluates $(k+1)(D+1)$ states for each of the $O(k D^{k+1})$ guesses, yielding a running time of $O(n \log n + k^2 D^{k+2})$, where the polynomial term accounts for the initial sorting and prefix sums. This establishes that the problem is FPT parameterized by $D+k$.

\begin{algorithm}[htb]
\caption{FPT algorithm for $k$-sum lateness parameterized by $D+k$}
\label{alg:fpt-Dk-dp}
\DontPrintSemicolon
Sort the distinct due dates as $\delta_1<\cdots<\delta_D$\;
\ForEach{$i\in[D]$}{
  Compute $J_i=\{j\in\Jobs:d_j=\delta_i\}$ and sort by $\prec_p$\;
  Compute $A_i(t)=\sum_{s=1}^t p_{j_{i,s}}$, for $0\le t\le |J_i|$\;
}
$v_{\mathrm{best}}\leftarrow+\infty$\;
\ForEach{$\sigma\in[D]^k$}{
  Compute $c_i(\sigma)$ for all $i\in[D]$\;
  \lIf{$c_i(\sigma)>|J_i|$ for some $i$}{discard $\sigma$ and continue}
  Define $X_i(\sigma)$, $B_i(\sigma)$, and the ordered jobs $x_1,\ldots,x_k$\;
  Compute $P_X(a)$, $0\le a\le k$, and $P_B(h)$, $0\le h\le D$\;
  \ForEach{$(q,h_0)$ with $q\in[k]$ and $h_0\in\{0,\ldots,D\}$}{
    $\theta\leftarrow P_X(q)+P_B(h_0)-d_{x_q}$\;
    Set all entries of $M[0..k,0..D]$ to $+\infty$ and then update $M[0,0]\leftarrow 0$\;
    \ForEach{$(a,h)$ with $0\le a\le k$, $0\le h\le D$, in nondecreasing order of $a+h$}{
      \If{$M[a,h]<+\infty$}{
        \If{$h<D$ and $(a\ge q$ or $h<h_0)$}{
          \lIf{$B_{h+1}=\emptyset$ or $P_X(a)+P_B(h+1)-\delta_{h+1}\le \theta$}{$M[a,h+1]\leftarrow \min\{M[a,h+1],M[a,h]\}$}
        }
        \If{$a<k$}{
          $\ell_X\leftarrow P_X(a+1)+P_B(h)-d_{x_{a+1}}$\;
          \lIf{$(a+1=q$ and $h=h_0)$ or $(a+1\ne q$ and $\ell_X\ge \theta)$}{$M[a+1,h]\leftarrow \min\{M[a+1,h],M[a,h]+\ell_X\}$}
        }
      }
    }
    $v_{\mathrm{best}}\leftarrow\min\{v_{\mathrm{best}},M[k,D]\}$\;
  }
}
\Return $v_{\mathrm{best}}$\;
\end{algorithm}

\begin{theorem}
\label{thm:fpt-kD-dp}
The optimal value of the $k$-sum lateness problem can be computed in time
\[
O(n\log n+k^2D^{k+2}).
\]
\end{theorem}

\begin{proof}
We first prove correctness. By Theorem~\ref{thm:block-island}, choose an
optimal schedule $S^*$ and a top-$k$ set $X^*$ satisfying the block--island
decomposition and the suffix property inside each due-date class. Let $\sigma^*\in[D]^k$ be the word of the due-date classes of the jobs of
$X^*$ in their order in $S^*$. Thus $\sigma^*_a=i$ if the $a$-th job of $X^*$
in $S^*$ belongs to $J_i$. The suffix property fixes the set $X^*\cap J_i$ in
each class $J_i$, and the increasing property fixes the order of these jobs inside the
class. Hence the word $\sigma^*$ reconstructs exactly the ordered list of jobs
from $X^*$. For this word, the sets $X_i(\sigma^*)$ are exactly $X^*\cap J_i$,
and the blocks $B_i(\sigma^*)$ are exactly $J_i\setminus X^*$.

By the block--island decomposition, the jobs outside $X^*$ appear as complete
blocks $B_1,\ldots,B_D$ in due-date order, interleaved with the jobs of $X^*$.
Empty blocks are included as zero-length block transitions. Thus $S^*$ corresponds
to a path in the DP for the word $\sigma^*$.

Choose $q$ so that $x_q$ has minimum lateness among the jobs of $X^*$ in
$S^*$. Let $h_0$ be the largest index such that the block transition $B_{h_0}$
is processed before $x_q$ on this path, with $h_0=0$ if no block transition is
processed before $x_q$. Then exactly the block transitions $B_1,\ldots,B_{h_0}$
are processed before $x_q$, and
\[
\theta=P_X(q)+P_B(h_0)-d_{x_q}=L_{x_q}(S^*).
\]
Since $x_q$ has minimum lateness inside $X^*$, every job of $X^*$ has
lateness at least $\theta$. Since $X^*$ is a top-$k$ set for $S^*$, every
job outside $X^*$ has lateness at most $\theta$.

It follows that every transition used by $S^*$ is allowed by the DP. The
transition that places $x_q$ is made exactly when $h=h_0$. Every other job
$x_a\in X^*$ has lateness at least $\theta$. Every nonempty block $B_i$
contains only jobs outside $X^*$, so all jobs in $B_i$ have lateness at most
$\theta$. The DP checks only the last job of $B_i$. This is enough because
all jobs in $B_i$ have due date $\delta_i$, and their completion times
inside the block are nondecreasing. This also holds when some processing times
are zero. Hence the block test is satisfied. Empty blocks impose no condition.

Along this path, the accumulated value is
\[
\sum_{x\in X^*} L_x(S^*)=\lambda_k(S^*).
\]
Therefore the algorithm returns a value at most the optimum.

Conversely, we prove that every finite DP path represents a real schedule whose
DP value is its $k$-sum lateness. Consider any finite DP path for some guess
$(\sigma,q,h_0)$. Expand the path into a schedule $S$ by replacing each block
transition for $B_i$ with the jobs of $B_i$ in their \SPT{} order, and each
job transition with the corresponding job $x_a$. This gives a schedule of all
jobs.

We claim that
\[
X(\sigma)=\{x_1,\ldots,x_k\}
\]
is a top-$k$ set for $S$. The DP places $x_q$ with lateness exactly
$\theta$. It also requires every other job $x_a$ to have lateness at least
$\theta$. Hence every job in $X(\sigma)$ has lateness at least $\theta$.

Now let $y\notin X(\sigma)$. Then $y$ belongs to some block $B_i$. If
$B_i$ is empty there is nothing to prove. Otherwise, all jobs in $B_i$ have
due date $\delta_i$, and their completion times inside the block are
nondecreasing. Thus the maximum lateness in $B_i$ is attained by the last job
of the block. The block transition checks that this last job has lateness at
most $\theta$. Therefore $L_y(S)\le \theta$.

Thus every job in $X(\sigma)$ has lateness at least $\theta$, and every job
outside $X(\sigma)$ has lateness at most $\theta$. Hence $X(\sigma)$ is a
top-$k$ set for $S$. The value accumulated by the DP is exactly
\[
\sum_{a=1}^k L_{x_a}(S).
\]
Since $X(\sigma)$ is a top-$k$ set for $S$, this sum is equal to
$\lambda_k(S)$. Therefore every finite DP value is the value of a real
schedule. No finite DP value is smaller than the optimum. Together with the
previous paragraph, this proves correctness.

We now bound the running time. Sorting all jobs by due date and sorting each
due-date class by \SPT{} takes $O(n\log n)$ time. The prefix sums $A_i$
are computed in $O(n)$ time.

There are $D^k$ words $\sigma$. For one word, the values $c_i(\sigma)$, the
ordered list $x_1,\ldots,x_k$, and the arrays $P_X$ and $P_B$ are
computed in $O(k+D)$ time from the sorted lists and prefix sums. For this
word, there are $k(D+1)$ choices of $(q,h_0)$. For each such choice, the DP has
$(k+1)(D+1)=O(kD)$ states. Each state has at most two transitions, and each
transition is evaluated in constant time. Hence one DP takes $O(kD)$ time.

The total time after preprocessing is $D^k\cdot k(D+1)\cdot O(kD)=O(k^2D^{k+2})$.
This proves the claimed bound for the optimal value.

To output a schedule, store parent pointers for the best DP run, or rerun that
run with parent pointers. The DP path has length $O(k+D)$. Expanding the
block transitions writes each job once. This adds $O(n)$ time.
\end{proof}

\begin{corollary}
\label{cor:fpt-tardiness}
The $k$-sum tardiness problem can be solved in time
\[
O(n\log n+k^2(D+1)^{k+2})
\]
where $D$ is the number of distinct due dates in the original tardiness instance.
\end{corollary}

\begin{proof}
Given an instance of $k$-sum tardiness on $n$ jobs with $D$ distinct due dates,
we apply the reduction from Proposition~\ref{prop:zero-dummy-reduction}. We add
$k$ dummy jobs with processing time $0$ and due date $0$. This produces an
equivalent instance of $k$-sum lateness with $n+k$ jobs and at most $D+1$
distinct due dates. By Theorem~\ref{thm:fpt-kD-dp}, this augmented instance can
be solved in time
\[
O((n+k)\log(n+k)+k^2(D+1)^{k+2}).
\]
Since $k\le n$, this is
\[
O(n\log n+k^2(D+1)^{k+2}).
\]
Restricting the returned schedule to the original jobs gives an optimal
schedule for the original tardiness instance.
\end{proof}

\section{Conclusions, open problems, and the \EDD{}--\SPT{} asymmetry}
    \label{sec:concl}
    
    We presented XP algorithms for fixed $k$, fixed $D$, and fixed $P$, a
    pseudopolynomial algorithm for integral data, and an FPT algorithm parameterized by $D+k$.
    For unrestricted $k$, the problem is weakly NP-complete.

    It remains open whether the dependence on a constant number of due dates can be improved
    substantially, or whether the problem is fixed-parameter tractable with respect to $D$ alone.
    One possible route would be a mixed-integer linear formulation with $O(D)$ integer variables
    and polynomially many continuous variables and constraints.
    Such a formulation would imply an $f(D)\operatorname{poly}(n)$ algorithm by Lenstra's
    theorem for integer programming in fixed dimension \cite{Lenstra1983}.

\paragraph{Asymmetry between small and large $k$.}
A natural question is whether the small-$k$ and large-$k$ regimes behave symmetrically.
For $k=1$, the \EDD{} order is optimal, while for $k=n$, the \SPT{} order is optimal.
For fixed (small) $k$, one can use that some optimal schedule is close to \EDD{}.
This suggests that for small $h=n-k$, some optimal schedule might be close to \SPT{}.
The example below shows this is false: even for $h=1$, the unique optimal
schedule can be at adjacency distance $n-1$ from \SPT{}. Hence any algorithm
for small $h$ needs a different structural principle.

\begin{proposition}[No \SPT{} proximity for $h=1$]
\label{prop:no-spt-proximity-h1}
Fix $n\ge 2$ and let $k=n-1$ (equivalently $h=n-k=1$). There is an instance
whose unique optimal schedule has adjacency distance $n-1$ from the \SPT{}
order.
\end{proposition}

\begin{proof}
Set $\varepsilon:=1/(4n^3)$ and $p_j:=1+\varepsilon j$ for $j=1,\ldots,n$. Then the \SPT{}
order is
\[
S_{\mathrm{spt}}=(1,2,\ldots,n),
\]
and the reverse \SPT{} order is
\[
S_{\mathrm{rev}}=(n,n-1,\ldots,1).
\]
Define due dates by $d_j:=C_j(S_{\mathrm{rev}})$. Then $L_j(S_{\mathrm{rev}})=0$ for all
$j$, so $\lambda_{n-1}(S_{\mathrm{rev}})=0$.

We claim that $S_{\mathrm{rev}}$ is the unique optimum. Take any schedule
$S\neq S_{\mathrm{rev}}$ and set
$\Delta_j:=C_j(S)-C_j(S_{\mathrm{rev}})$. Since
$d_j=C_j(S_{\mathrm{rev}})$, we have $L_j(S)=\Delta_j$ and hence
\[
\lambda_{n-1}(S)=\sum_{j=1}^n \Delta_j-\min_j \Delta_j.
\]

First consider $\varepsilon=0$ (unit processing times). Then
$\Delta_j=\operatorname{pos}_S(j)-\operatorname{pos}_{S_{\mathrm{rev}}}(j)$ and
$\sum_j\Delta_j=0$. Since $S\neq S_{\mathrm{rev}}$, not all $\Delta_j$ are zero,
so $\min_j\Delta_j\le -1$ and therefore
\[
\sum_j\Delta_j-\min_j\Delta_j\ge 1.
\]

Now restore $\varepsilon=1/(4n^3)$. For any schedule $T$, the change in
$C_j(T)$ from the unit-processing-time instance to the instance
$p_i=1+\varepsilon i$ is $
\varepsilon\cdot \sum_{i:\operatorname{pos}_T(i)\le \operatorname{pos}_T(j)} i,
$
which lies in $[0,B]$, where
$
B:=\varepsilon\frac{n(n+1)}2$.
Hence, for every $j$, the change in $
\Delta_j=C_j(S)-C_j(S_{\mathrm{rev}})$
is the difference of two numbers in $[0,B]$, and therefore has absolute value
at most $B$. It follows that $\sum_j\Delta_j$ changes by at most $nB$ and that
$\min_j\Delta_j$ changes by at most $B$. Therefore the quantity
$\sum_j\Delta_j-\min_j\Delta_j$ changes by at most
\[
(n+1)B=(n+1)\varepsilon\frac{n(n+1)}2<1.
\]
Thus for every $S\neq S_{\mathrm{rev}}$ we still have
$\lambda_{n-1}(S)>0$, so $S_{\mathrm{rev}}$ is uniquely optimal.

Finally, $S_{\mathrm{spt}}=(1,\ldots,n)$ and
$S_{\mathrm{rev}}=(n,\ldots,1)$ share no common subsequence of length greater
than $1$, so their adjacency distance is $n-1$.
\end{proof}
\medskip
\paragraph*{Acknowledgements} This work was partially supported by ANID (Agencia Nacional de Investigaci\'on y Desarrollo, Chile) through FONDECYT Grant No.~1231669 and the BASAL Center for Mathematical Modeling (FB210005).
    \bibliographystyle{plain}
    \bibliography{refs}

\end{document}